\documentclass{IEEEtran}

\pdfoutput=1

\usepackage{amsthm}
\usepackage{amsmath,amssymb}
\usepackage[final,hyperfootnotes=false,hidelinks]{hyperref}

\usepackage{tikz}
\usepgflibrary{shapes.misc}
\usetikzlibrary{arrows,positioning,fit,backgrounds}
\usetikzlibrary{shapes.misc} 
\usetikzlibrary{calc}

\newtheorem{defn}{Definition}
\newtheorem{lem}{Lemma}
\newtheorem{thm}{Theorem}
\newtheorem{prop}[lem]{Proposition}
\newtheorem{cor}[lem]{Corollary}
\newtheorem{rem}{Remark}

\newtheorem{conj}{Conjecture}

\DeclareMathOperator{\vol}{vol}

\DeclareMathOperator{\tr}{tr}
\DeclareMathOperator{\var}{var}
\DeclareMathOperator{\supp}{supp}
\DeclareMathOperator*{\argmin}{arg\,min}

\newcommand{\uX}{\underline{X}}
\newcommand{\uY}{\underline{Y}}

\begin{document}

\title{A Few Interactions Improve Distributed Nonparametric Estimation, Optimally}

\author{%
Jingbo~Liu,~\IEEEmembership{Member,~IEEE}
\thanks{Jingbo Liu is with the Department of Statistics and the Department of Electrical and Computer Engineering,
   University of Illinois, Urbana-Champaign, IL, 61820, US. Email: jingbol@illinois.edu}%
\thanks{Manuscript received 16-Feb-2022; revised 27-Nov-2022. Associate editor: Himanshu Tyagi.}
\thanks{This paper was presented in part at 2022 IEEE International Symposium on Information Theory (ISIT) \cite{liu2022interaction}.}
}

\markboth{to appear on IEEE Trans. Inf. Theory
}%
{Shell \MakeLowercase{\textit{et al.}}: Bare Demo of IEEEtran.cls for IEEE Journals}

\maketitle

\begin{abstract}
Consider the problem of nonparametric estimation of an unknown $\beta$-H\"older smooth density $p_{XY}$ at a given point, where $X$ and $Y$ are both $d$ dimensional.
An infinite sequence of i.i.d.\ samples $(X_i,Y_i)$ are generated according to this distribution, and two terminals observe $(X_i)$ and $(Y_i)$, respectively.
They are allowed to exchange $k$ bits either in oneway or interactively in order for Bob to estimate the unknown density.
We show that the minimax mean square risk is order $\left(\frac{k}{\log k}
\right)^{-\frac{2\beta}{d+2\beta}}$ for one-way protocols and 
$k^{-\frac{2\beta}{d+2\beta}}$ for interactive protocols. 
The logarithmic improvement is nonexistent in the parametric counterparts,
and therefore can be regarded as a consequence of nonparametric nature of the problem.
Moreover, a few rounds of interactions  achieve the interactive minimax rate: the number of rounds can grow as slowly as the super-logarithm (i.e., inverse tetration) of $k$.
The proof of the upper bound is based on a novel multi-round scheme for estimating the joint distribution of a pair of biased Bernoulli variables,
and the lower bound is built on a sharp estimate of a symmetric strong data processing constant for biased Bernoulli variables.
\end{abstract}



\IEEEpeerreviewmaketitle

\section{Introduction}
The communication complexity problem  was introduced in the seminal paper of Yao \cite{yao1979some} (see also \cite{kushilevitz} for a survey), where two terminals (which we call Alice and Bob) compute a given Boolean function of their local inputs (${\bf X}=(X_i)_{i=1}^n$ and ${\bf Y}=(Y_i)_{i=1}^n$) by means of exchanging messages.
The famous log-rank conjecture
provides an estimate of the communication complexity of a general Boolean function, 
which is still open to date.
Meanwhile,
communication complexity of certain specific functions can be better understood.
For example,
the Gap-Hamming problem
\cite{Indyk}\cite{cr12}
concerns testing $f({\bf X,Y})>\frac1{n}$ against $f({\bf X,Y})<-\frac1{n}$,
where 
$f({\bf X,Y}):=\frac1{n}\sum_{i=1}^nX_iY_i$
denotes the sample correlation and $X_i,Y_i\in\{+1,-1\}$.
It was shown in \cite{cr12} with a geometric argument that the communication complexity (for worst-case deterministic ${\bf X,Y}$) is $\Theta(n)$,
therefore a one-way protocol where Alice  simply sends ${\bf X}$ cannot be improved (up to a multiplicative constant) by an interactive protocol.

Gap-Hamming is closely related to the problem of estimating the joint distribution of a pair of binary or Gaussian random variables (using $n$ i.i.d.\ samples).
Indeed, for $n$ large we may assume that Alice (resp.\ Bob) can estimate the marginal distributions of $X$ (resp.\ $Y$) very well, 
so that the joint distribution is parameterized by only one scalar which is the correlation.
An information-theoretic proof of Gap-Hamming was previously provided in \cite{hadar2019communication}, 
building on a converse for correlation estimation for the binary symmetric distribution,
and pinned down the exact prefactor in the risk-communication tradeoff. 
In particular, the result of \cite{hadar2019communication} implies that 
the naive algorithm where Alice simply sends $\bf X$ can be improved by a constant factor in the estimation risk
by a more sophisticated scheme using additional samples.
For the closely related problem of correlation (distribution) testing,
\cite{8437584} and 
\cite{xiang2013interactive}
provided asymptotically tight bounds on the communication complexity under the one-way and interactive protocols
when the null hypothesis is the independent distribution (zero correlation),
which also implies that the error exponent can be improved by an algorithm using additional samples.
The technique of \cite{xiang2013interactive} is based on the tensorization of internal and external information (\eqref{e_sstartens} ahead),
whereas the bound of \cite{8437584} uses hypercontractivity.
More recently, \cite{8849426} derived bounds for testing against \emph{dependent} distributions using optimal transport inequalities.

In this paper, we take the natural step of introducing nonparametric (NP) statistics to Alice and Bob,
whereby two parties estimate a nonparametric density by means of sending messages interactively.
It will be seen that this problem is closely related to a ``sparse'' version of the aforementioned Gap-Hamming problem, 
where interaction does help,  
in contrast to the usual Gap-Hamming problem.

For concreteness, consider the problem of nonparametric estimation of an unknown $\beta$-H\"older smooth density $p_{XY}$ at a given point $(x_0,y_0)$. 
For simplicity we assume the symmetric case where $X$ and $Y$ are both $d$ dimensional.
An infinite sequence of i.i.d.\ samples $(X_i,Y_i)$ are generated according to $p_{XY}$, and Alice and Bob observe $(X_i)$ and $(Y_i)$, respectively.
After they exchange $k$ bits (either in one-way or interactively), 
Bob estimates the unknown density at the given point.
We successfully characterize the minimax rate in terms of the communication complexity $k$:
it is order $\left(\frac{k}{\log k}
\right)^{-\frac{2\beta}{d+2\beta}}$ for one-way protocols and 
$k^{-\frac{2\beta}{d+2\beta}}$ for interactive protocols. 

Notably, allowing interaction strictly improves the estimation risk.
Previously, separations between one-way and interactive protocols are known but in very different contexts: In \cite[Corollary~1]{liu2017secret} (see also \cite{liu2016key}), the separation was found in the rate region of common randomness generation from biased binary distributions, using certain convexity arguments,
but this only implies a difference in the leading constant, rather than the asymptotic scaling. 
On the other hand, the example distribution in 
\cite{sudan2019communication_2}
is based on the pointer-chasing construction of \cite{nisan1991rounds}, which appears to be a highly artificial distribution designed to entail a separation between the one-way and interactive protocols.
Another example where interaction improves zero-error source coding with side information, based on a ``bit location'' algorithm, was described in
\cite{orlitsky1990worst},
and it was shown that two-way communication complexity differs from interactivity communication complexity only by constant factors.
In contrast, the logarithmic separation in the present paper arises from the nonparametric nature of the problem:
If we consider the problem of correlation estimation for Bernoulli pairs with a \emph{fixed} bias
(a parametric problem),
the risk will be order $k^{-\frac1{2}}$, 
and there will be no separation between one-way and interactive protocols (which is indeed the case in \cite{hadar2019communication}).
In contrast, nonparametric estimation is analogous to Bernoulli correlation estimation where the bias changes with $k$ (since the optimal bandwidth adapts to $k$),
which gives rise to the separation.

For the risk upper bound, in the one-way setting it is efficient for Alice to just encode the set of $i$'s such that $X_i$ falls within a neighborhood (computed by the optimal bandwidth for a given $k$) of the given point $x_0$.
To achieve the optimal $k^{-\frac{2\beta}{d+2\beta}}$ rate for interactive protocols,
we provide a novel scheme that uses $r>1$ rounds of interactions, where $r=r(k)$ grows as slowly as the super logarithm (i.e.\ the inverse of tetration) of $k$.
With the sequence of $r(k)$ we use in Section~\ref{sec_rinfty} (and suppose that $\beta=d=1$), while $r=4$ interactions is barely enough for $k$ equal to the number of letters in a short sentence,
$r=8$ is more than sufficient for $k$ equal to the number of all elementary particles in the entire observable universe.
Thus from a practical perspective, $r(k)$ is effectively a constant, although it remains an interesting theoretical question whether $r(k)$ really diverges (Conjecture~\ref{conj1}).

For the lower bound,
the proof is based on the \emph{symmetric data processing constant} introduced in \cite{liu2017secret}. 
Previously, the data processing constant $s^*_r$
has been connected to two-party estimation and hypothesis testing in \cite{hadar2019communication};
the idea was 
canonized as the following statement:
``Information for hypothesis testing locally'' is upper bounded by $s^*_r$ times ``Information communicated mutually''.
However, $s^*_r$ is not easy to compute, and previous bounds on $s^*_r$ are also not tight enough for our purpose.
Instead, we first use an idea of simulation of continuous variables to reduce the problem to estimation of Bernoulli distributions, 
for which $s^*_r$ is easier to analyze.
Then we use some new arguments to bound $s^*_{\infty}$.

Let us emphasize that this paper concerns density estimation at a given point, rather than estimating the global density function.
For the latter problem, it is optimal for Alice to just quantize the samples and send it to Bob, which we show in the companion paper \cite{liu2022_ciss} .
The mean square error (in $\ell_2$ norm) of estimating global density function scales differently than the case of point-wise density estimation since the messages cannot be tailored to the given point.

{\bf Related work. }
Besides function computation, 
distribution estimation and testing, 
other problems which have been studied in the communication complexity or privacy settings include lossy source coding
\cite{kaspi} and common randomness or secret key generation \cite{tyagi2013common}\cite{liu2016smoothing}\cite{liu2017secret}\cite{sudan2019communication}.
The key technical tool for interactive two-way communication models, 
namely the tensorization of internal and external information (\eqref{e_sstartens} ahead), 
appeared in \cite{kaspi} for lossy compression,
\cite{br11}\cite{ma2011some} for function computation,
\cite{tyagi2013common}\cite{liu2017secret} for common randomness generation,
and 
\cite{xiang2013interactive}\cite{hadar2019communication} for parameter estimation.

For one-way communication models, 
the main tool is a tensorization property related to the strong data processing constant (see \eqref{e_tens} ahead),
which was first used in \cite{ac86}
in the study of the error exponents in communication constrained hypothesis testing.
The hypercontractivity method for single-shot bounds in one-way models was used in \cite{liu2016smoothing}\cite{liu_thesis} for common randomness generation and \cite{8437584} for testing.

In statistics, communication-constrained estimation has received considerable attention recently, starting from \cite{zhang2013information},
which considered a model where distributed samples are compressed and sent to a central estimator.
Further works on this communication model include 
settings of Gaussian location estimation 
\cite{braverman2016communication}\cite{cai2020distributed},
parametric estimation
\cite{han2018geometric},
nonparametric regression
\cite{zhu2018distributed},
Gaussian noise model
\cite{zhu2018quantized},
statistical inference 
\cite{acharya2020inference},
and
nonparametric estimation
\cite{han2018distributed}(with a bug fixed in \cite{barnes2020lower})
\cite{acharya2021optimal}.
Related problems solved using similar techniques include differential privacy
\cite{duchi2013local}
and data summarization \cite{phillips2020near}\cite{turner2021statistical}\cite{turner2021efficient}.
Communication-efficient construction of test statistics for distributed testing using the divide-and-conquer algorithm is studied in
\cite{battey2018distributed}.
Generally speaking, these works on statistical minimax rates concern the so-called \emph{horizontal} partitioning of data sets, where data sets share the same feature space but differ in samples \cite{yang2019federated}\cite{kairouz2021advances}.
In contrast, \emph{vertical} distributed or federated learning, where data sets differ in features, has been used by corporations such as those in finance and medical care \cite{yang2019federated}\cite{kairouz2021advances}. 
It is worth mentioning that such horizontal partitioning model was also introduced in Yao's paper \cite{yao1979some}  in the context of function computation under the name ``simultaneous message model'', where different parties send messages to a referee instead of to each other.
The direct sum property (similar to the tensorization property of internal and external information) of the simultaneous message model was discussed in \cite{chakrabarti2001informational}.

{\bf Organization of the paper. }
We review the background on nonparametric estimation, data processing constants and testing independence in Section~\ref{sec_prelim}.
The formulation of the two-party nonparametric estimation problem and the summary of main results are given in Section~\ref{sec_main}.
Section~\ref{sec_ipu} examines the problem of estimating a parameter in a pair of biased Bernoulli distributions, which will be used as a building block in our nonparametric estimation algorithm.
Section~\ref{sec_exchanges} proves some bounds on information exchanges, 
which will be the key auxiliary results for the proof of upper bound for Bernoulli estimation in Section~\ref{app4}, and for nonparametric estimation in
Section~\ref{app2}.
Finally, lower bounds are proved in Section~\ref{sec_1pl} in the one-way case and in Section~\ref{app9} in the interactive case.

\section{Preliminaries}\label{sec_prelim}
\subsection{Notation}
We use capital letters for probability measures and lower cases for the densities functions.
We use the abbreviations $U_i^j:=(U_i,\dots,U_j)$ and $U^j:=U_1^j$.
We use boldface letters to denote vectors, for example
${\bf U}_i=(U_i(l))_{l=1}^n$.
Unless otherwise specified, the base of logarithms can be arbitrary but remain consistent throughout equations.
The precise meaning of the Landau notations, such as $O(\cdot)$, will be explained in each section or proofs of specific theorems.
We use
$\sum_{1\le i\le r}^{\rm odd}$ to denote summing over $i\in\{1,\dots,r\}\setminus 2\mathbb{Z}$.
For the vector representation of a binary probability distribution, we use the convention that $P_U=[P_U(0),P_U(1)]$.
For the matrix representation of a joint distribution of a pair of binary random variables, we use the convention that $P_{XY}=\begin{bmatrix}
P_{XY}(0,0) & P_{XY}(0,1)
\\
P_{XY}(1,0) & P_{XY}(1,1)
\end{bmatrix}
$.
For $x\in[0,1]$, we use the shorthand $\bar{x}:=1-x$.

\subsection{Nonparametric Estimation}\label{sec_npestimation}
Let us recall the basics about the problem of estimating a smooth density;
more details may be found in \cite{tsybakov2008introduction,stone1980optimal}.
Let $d\ge 1$ be an integer, and $s=(s_1,\dots,s_d)\in\{0,1,2,\dots\}^d$ be a multi-index.
For $x=(x_1,\dots,x_d)\in\mathbb{R}^d$,
let $D^s$ denote the differential operator 
\begin{align}
D^s=\frac{\partial^{s_1+\dots+s_d}}{\partial x_1^{s_1}\cdots \partial x_d^{s_d}}.
\end{align}
Given $\beta\in(0,\infty)$, let $\lfloor \beta\rfloor$ be the maximum integer \emph{strictly} smaller than $\beta$ \cite{tsybakov2008introduction} (note the difference with the usual conventions). 
Given a function $f$ whose domain includes a set $\mathcal{A}\subseteq\mathbb{R}^d$, define $\|f\|_{\mathcal{A},\beta}$ as the minimum $L\ge 0$ such that 
\begin{align}
|D^sf(x_1)-D^sf(x_2)|\le L\|x_1-x_2\|_2^{\beta-\lfloor \beta\rfloor},
\quad\forall x_1,x_2\in\mathcal{A},
\label{e_beta}
\end{align}
for all multi-indices $s$ such that $s_1+\dots+s_d=\lfloor \beta\rfloor$.
For example, $\beta=1$ define a Lipschitz function and an integer $\beta$ defines a function with bounded $\beta$-th derivative.

Given $L>0$, let $\mathcal{P}(\beta,L)$ be the class of probability density functions $p$ satisfying $\|p\|_{\mathbb{R}^d,\beta}\le L$.
Let $x_0\in\mathbb{R}^d$ be arbitrary.
The following result on the minimax estimation error is well-known:
\begin{align}
\inf_{T_n}\sup_{p\in\mathcal{P}(\beta,L)}\mathbb{E}[|T_n-p(x_0)|^2]
=\Theta(n^{-\frac{2\beta}{d+2\beta}})
\label{e3}
\end{align}
where the infimum is over all estimators $T_n$ of $p(x_0)$,
i.e., measurable maps from i.i.d.\ samples $X_1,\dots,X_n\sim p$ to $\mathbb{R}$.
$\Theta(\cdot)$ in \eqref{e3} may hide constants independent of $n$.

We say $K\colon \mathbb{R}^d\to\mathbb{R}$ is a kernel of order $l$ ($l\in\{1,2,\dots\}$) if $\int K=1$ and all up to the $l$-th derivatives of the Fourier transform of $K$ vanish at $0$ \cite[Definition~1.3]{tsybakov2008introduction}.
Therefore the rectangular kernel, which is the indicator of a set, is order 1.
A \emph{kernel estimator} has the form
\begin{align}
T_n=\frac1{nh^d}\sum_{l=1}^nK\left(\frac{X_l-x_0}{h}\right)
\label{e_kestimator}
\end{align}
where $h\in(0,\infty)$ is called \emph{bandwidth}.
If $K$ is a kernel of order $l=\lfloor \beta\rfloor$, then the kernel estimator \eqref{e_kestimator} with appropriate $h$ achieves the bound in \eqref{e3}
\cite[Chapter~1]{tsybakov2008introduction}.
In particular, the rectangular kernel is minimax optimal for $\beta\in(0,2]$.

If $K$ is compactly supported, then only local smoothness is needed, and density lower bound does not change the rate: we have
\begin{align}
\inf_{T_n}\sup_{p\in\mathcal{P}_{\mathcal{S}}(\beta,L,A)}\mathbb{E}[|T_n-p(x_0)|^2]
=\Theta(n^{-\frac{2\beta}{d+2\beta}})
\label{e_local}
\end{align}
where $\mathcal{S}$ is any compact neighborhood of $x_0$, $A\in[0,\frac1{\vol(\mathcal{S})})$ is arbitrary (with $\vol(\mathcal{S})$ denoting the volume of $\mathcal{S}$), and $\mathcal{P}_{\mathcal{S}}(\beta,L,A)$ denotes the non-empty set of probability density functions $p$ satisfying $\|p\|_{\mathcal{S},\beta}\le L$ and $\inf_{x\in\mathcal{S}}p(x)\ge A$.

\subsection{Strong and Symmetric Data Processing Constants}\label{sec_sdpi}
The strong data processing constant has proved useful in many distributed estimation problems \cite{braverman2016communication,ac86,duchi2013local,zhang2013information}.
In particular, it is strongly connected to two-party hypothesis testing under the one-way protocol.
In contrast, the symmetric data processing constant \cite{liu2017secret} can be viewed as a natural extension to interactive protocols.
This section recalls their definitions and auxiliary results, which will mainly be used in the proofs of lower bounds; 
however, the intuitions are useful for the upper bounds as well.

Given two probability measures $P$, $Q$ on the same measurable space, define the KL divergence
\begin{align}
D(P\|Q):=\int\log \left(\frac{dP}{dQ}\right) dP.
\end{align}
Define the $\chi^2$-divergence
\begin{align}
D_{\chi^2}(P\|Q):=\int \left(\frac{dP}{dQ}-1\right)^2dQ.
\label{e_chi2}
\end{align}
Let $X,Y$ be two random variables with joint distribution $P_{XY}$. Define the mutual information
\begin{align}
I(X;Y):=D(P_{XY}\|P_X\times P_Y).
\end{align}

\begin{defn}
Let $P_{XY}$ be an arbitrary distribution on $\mathcal{X}\times \mathcal{Y}$.
Define the strong data processing constant 
\begin{align}
s^*(X;Y):=\sup_{P_{U|X}}\frac{I(U;Y)}{I(U;X)}
\end{align}
where $P_{U|X}$ is a conditional distribution (with $\mathcal{U}$ being an arbitrary set),
and the mutual informations are computed under the joint distribution $P_{U|X}P_{XY}$.
\end{defn}
Clearly, the value of $s^*(X;Y)$ does not depend on the choice of the base of logarithm.
A basic yet useful property of the strong data processing constant is \emph{tensorization}:
if ${\bf (X,Y)}\sim P_{XY}^{\otimes n}$ then 
\begin{align}
s^*({\bf X;Y})=s^*(X;Y).
\label{e_tens}
\end{align}
Now if $({\bf X;Y})$ are the samples observed by Alice and Bob, $\Pi_1$ denotes the message sent to Bob, then $I(\Pi_1;{\bf X})\le k$ implies that
\begin{align}
D(P_{\Pi_1 {\bf Y}}\|P_{\Pi_1} P_{\bf Y})
\le s^*(X;Y)k.
\label{e11}
\end{align}
The left side is the KL divergence between the distribution under the hypothesis that $(X,Y)$ follows some joint distribution, and the distribution under the hypothesis that $X$ and $Y$ are independent.
Thus the error probabilities in testing against independence with \emph{one-way protocols} can be lower bounded.
This simple argument dates back at least to \cite{ac86,ahlswede1990minimax}.

A similar argument can be extended to testing independence under \emph{interactive protocols} \cite{xiang2013interactive}.
The fundamental fact enabling such extensions is the tensorization of certain information-theoretic quantities,
which appeared in various contexts \cite{kaspi,br11,liu2017secret}.

\begin{defn}\label{defn2}
Let $(X,Y)\sim P_{XY}$.
For given $r<\infty$,
define $s_r^*(X;Y)$ as the supremum of $R/S$ such that there exists random variables $U_1,\dots,U_r$ satisfying 
\begin{align}
R\le \sum_{1\le i\le r}^{\rm odd}I(U_i;Y|U^{i-1})
+\sum_{1\le i\le r}^{\rm even}I(U_i;X|U^{i-1});
\label{e_R}
\\
S\ge \sum_{1\le i\le r}^{\rm odd}I(U_i;X|U^{i-1})
+\sum_{1\le i\le r}^{\rm even}I(U_i;Y|U^{i-1}),
\label{e_S}
\end{align}
and
\begin{align}
U_i-(X,U^{i-1})-Y, \quad i\in\{1,\dots,r\}\setminus 2\mathbb{Z}
\label{e_markov0}
\\
U_i-(Y,U^{i-1})-X, \quad i\in\{1,\dots,r\}\cap 2\mathbb{Z}
\label{e_markov}
\end{align}
are Markov chains.
We call $s_{\infty}^*(X;Y)$ the symmetric data processing constant.
\end{defn}
Let us remark that using the Markov chains we have 
the right side of \eqref{e_R}
\begin{align}
 &\quad\sum_{1\le i\le r}^{\rm odd}I(U_i;Y|U^{i-1})
+\sum_{1\le i\le r}^{\rm even}I(U_i;X|U^{i-1})
\nonumber\\
&=I(X;Y)-I(X;Y|U^r)
\\
&=I(U^r;XY)-[I(U^r;X|Y)+I(U^r;Y|X)]
\end{align}
whereas the right side of \eqref{e_S}
 \begin{align}
 \sum_{1\le i\le r}^{\rm odd}I(U_i;X|U^{i-1})
+\sum_{1\le i\le r}^{\rm even}I(U_i;Y|U^{i-1})
=I(U^r;XY).
\label{e_extrinsic}
 \end{align}
 In the computer science literature \cite{br11}, $I(U^r;XY)$ is called the \emph{external information} whereas $I(U^r;X|Y)+I(U^r;Y|X)$ the \emph{internal information}.

The symmetric strong data processing constant is symmetric in the sense that $s_{\infty}^*(X;Y)=s_{\infty}^*(Y;X)$, since $r=\infty$ in the definition.
On the other hand, $s_1^*(X;Y)$ coincides with the strong data processing constant which is generally not symmetric.
Furthermore, a tensorization property holds for the internal and external information: denote by $\mathcal{R}(X;Y)$ the set of all $(R,S)$ satisfying \eqref{e_R} and \eqref{e_S} for some $U_1,\dots,U_r$. 
Let ${\bf (X,Y)}\sim P_{XY}^{\otimes n}$.
Then 
\begin{align}
\mathcal{R}({\bf X;Y})=n\mathcal{R}(X;Y).
\end{align}
In particular, taking 
the slope of the boundary at the original yields
\begin{align}
s_{\infty}^*({\bf X;Y})
=s_{\infty}^*(X;Y).
\label{e_sstartens}
\end{align}

A useful and general upper bound on $s_{\infty}^*$ in terms of SVD was provided in \cite[Theorem~4]{liu2017secret}, which implies that $s_{\infty}^*=s_1^*$ when $X$ and $Y$ are unbiased Bernoulli.
However, that bound is not tight enough
for the nonparametric estimation problem we consider, 
and in fact we adopt a new approach in Section~\ref{app9} for the \emph{biased} Bernoulli distribution.
Let us remark that $s_{\infty}^*=s_1^*$ holds also for Gaussian $(X,Y)$, which follows by combining the result on unbiased Bernoulli distribution and a central limit theorem argument \cite{liu2017secret} (see also \cite{hadar2019communication}).
Moreover, it was conjectured in \cite{liu2017secret} that the set of possible $(R,S)$ satisfying \eqref{e_R}-\eqref{e_S} does not depend on $r$ when $X$ and $Y$ are unbiased Bernoulli.

\subsection{Testing Against Independence}
Consider the following setting: $P_{XY}$ is an arbitrary distribution on $\mathcal{X}\times \mathcal{Y}$; $P_{\bf XY}:=P_{XY}^{\otimes n}$; 
$\Pi=(\Pi_0,\dots,\Pi_r)$ is a sequence of random variables, with
$P_{\Pi|{\bf XY}}$ being given and satisfying
$P_{\Pi_0|{\bf XY}}=P_{\Pi_0}$, 
$P_{\Pi_i|{\bf XY}\Pi_0^{i-1}}=P_{\Pi_i|{\bf X}\Pi_0^{i-1}}$ for $i\in\{1,\dots,r\}\setminus 2\mathbb{Z}$ and $P_{\Pi_i|{\bf XY}\Pi_0^{i-1}}=P_{\Pi_i|{\bf Y}\Pi_0^{i-1}}$ for $i\in\{1,\dots,r\}\cap2\mathbb{Z}$;
$\bar{P}_{\bf XY}=P_{\bf X}P_{\bf Y}$ is the 
distribution under the hypothesis of independence, and 
 $\bar{P}_{\Pi{\bf XY}}:=P_{\Pi|{\bf XY}}\bar{P}_{\bf XY}$.
The following result is known in \cite{xiang2013interactive,8849426,hadar2019communication}:
\begin{lem}
$D(P_{{\bf Y}\Pi}\|\bar{P}_{{\bf Y}\Pi})\le I({\bf X;Y})-I({\bf X;Y}|\Pi)$.
\end{lem}

Now by Definition~\ref{defn2}, we immediately have
\begin{align}
s_r^*(X;Y)&\ge 
\frac{I({\bf X;Y})-I({\bf X;Y}|\Pi)}{I({\bf XY};\Pi)}
\\
&\ge 
\frac{D(P_{{\bf Y}\Pi}\|\bar{P}_{{\bf Y}\Pi})}{H(\Pi)}
\label{e19}
\end{align}
which generalizes \eqref{e11}.
Therefore, $s_r^*(X;Y)$ can be used to bound $D(P_{{\bf Y}\Pi}\|\bar{P}_{{\bf Y}\Pi})$, and in turn, the error probability in indepedence testing.

\section{Problem Setup and Main Results}
\label{sec_main}
We consider estimating the density function at a given point, where the density is assumed to be H\"older continuous in a neighborhood of that point.
It is clear that there is no loss of generality assuming such neighborhood to be the unit cube, and that the given point is its center.
More precisely, the class of densities under consideration is defined as follows:
\begin{defn}\label{defn_h}
Given $d\in\{1,2,\dots\}$, $L>0$, $A\in [0,1)$, and $\beta>0$,
let $\mathcal{H}(\beta,L,A)$ be the set of all probability density $p_{XY}$ on $\mathcal{X}\times \mathcal{Y}$ (where $\mathcal{X}=\mathcal{Y}=\mathbb{R}^d$) satisfying
\begin{align}
p_X(x),p_Y(y)\ge A,\quad \forall x,y\in[0,1]^d,
\label{e20}
\end{align}
and
\begin{align}
\|p_{XY}\|_{[0,1]^{2d},\beta}\le L.
\label{e_holder}
\end{align}
\end{defn}
\begin{defn}
We say $\mathcal{C}$ is a \emph{prefix code} \cite{elements} if it is a subset of the set of all finite non-empty binary sequences satisfying the property that for any distinct $s_1,s_2\in\mathcal{C}$, $s_1$ cannot be a prefix of $s_2$.
\end{defn}
The problem is to estimate the density at a given point of an unknown distribution from $\mathcal{H}(\beta,L,A)$.
More precisely,
\begin{itemize}
\item $P_{XY}$ is a fixed but unknown distribution whose corresponding density $p_{XY}$ belongs to $\mathcal{H}(\beta,L,A)$ for some $\beta\in(0,\infty)$, $L\in(0,\infty)$, and $A\in[0,1)$. 
\item Infinite sequence of pairs $(X(1),Y(1))$, $(X(2),Y(2))$,\dots are i.i.d.\ according to $P_{XY}$.
Alice (Terminal 1) observes ${\bf X}=(X(l))_{l=1}^{\infty}$ and 
Bob (Terminal 2) observes ${\bf Y}=(Y(l))_{l=1}^{\infty}$.
\item Unlimited common randomness $\Pi_0$ is observed by both Alice and Bob.
That is, an infinite random bit string independent of $\bf (X,Y)$ shared by Alice and Bob.
\item For $i=1,\dots,r$ ($r$ is an integer), if $i$ is odd, then Alice sends to Bob a message $\Pi_i$, which is an element in a prefix code, 
where $\Pi_i$ is computed using the common randomness $\Pi_0$, the previous transcripts $\Pi^{i-1}=(\Pi_1,\dots,\Pi_{i-1})$, and 
${\bf X}$;
if $i$ is even, then Bob sends to Alice a message $\Pi_i$ computed using $\Pi_0$, $\Pi^{i-1}$, and ${\bf Y}$.
\item Bob computes an estimate $\hat{p}$ of the true density $p_{XY}(x_0,y_0)$, where $x_0=y_0$ is the center of $[0,1]^d$.
\end{itemize}

{\bf One-way NP Estimation Problem}. Suppose that $r=1$. 
Under the constraint on the expected length of the transcript (i.e.\ length of the bit string)
\begin{align}
\mathbb{E}[|\Pi^r|]\le k,
\label{e_comm}
\end{align}
where $k>0$ is a real number,
what is the minimax risk
\begin{align}
R(k):=\min_{\hat{p},\bf \Pi}\max_{p_{XY}\in \mathcal{H}(\beta,L,A)}\mathbb{E}[|\hat{p}-p_{XY}(x_0,y_0)|^2]?
\end{align}
\noindent
{\bf Interactive NP Estimation Problem}.
Under the same constraint on the expected length of the transcript, but without any constraint on the number of rounds $r$, what is the minimax risk?
\begin{rem}
The prefix condition ensures that Bob knows that the current round has terminated after finishing reading each $\Pi_i$.
Alternatively, the problem can be formulated by stating that $\Pi_i$ is a random variable in an arbitrary alphabet, and replacing \eqref{e_comm} by the entropy constraint $H(\Pi^r)\le k$.
Furthermore, one may use the information leakage constraint $I({\bf X,Y};\Pi^r)\le k$ instead.
From our proofs it is clear that the minimax rates will not change under these alternative formulations.
\end{rem}
\begin{rem}
There would be no essential difference if the problem were formulated with $|\Pi|\le k$ almost surely and $|\hat{p}-p_{XY}(x_0,y_0)|^2\le R(k)$ with probability (say) at least $1/2$.
Indeed, for the upper bound direction, those conditions are satisfied with a truncation argument, once we have an algorithm satisfying $\mathbb{E}[|\Pi|]\le k/4$ and $\mathbb{E}[|\hat{p}-p_{XY}(x_0,y_0)|^2]\le R(k)/4$,
by Markov's inequality and the union bound,
therefore results only differ with a constant factor.
For the lower bound, the proof can be extended to the high probability version, since we used the Le Cam style argument \cite{assouad1996fano}.
\end{rem}

\begin{rem}
The common randomness assumption is common in the communication complexity literature, and, in some sense, is equivalent to private randomness \cite{newman1991private}.
In our upper bound proof, the common randomness is the randomness in the codebooks.
Random codebooks give rise to convenient properties, such as the fact that the expectation of the distribution of the matched codewords equals exactly the product of idealized single-letter distributions \eqref{e116}.
It is likely, however, that some approximate versions of these proofs steps, and ultimately the same asymptotic risk,
should hold for some carefully designed deterministic codebooks.
\end{rem}

\begin{thm}\label{thm3}
In one-way NP estimation, for any $\beta\in(0,\infty)$, $L\in(0,\infty)$, and $A\in[0,1)$, 
\begin{align}
R(k)=\Theta(\left(\frac{k}{\log k}
\right)^{-\frac{2\beta}{d+2\beta}})
\label{e_rpb3}
\end{align}
where $\Theta(\cdot)$ hides multiplicative factors depending on $L$, $\beta$, and $A$.
\end{thm}
The proof of the upper bound is in Section~\ref{sec_d1way}.
Recall that nonparametric density estimation using a rectangular kernel is equivalent to counting the frequency of samples in a neighborhood of a given diameter, the bandwidth, which we denote as $\Delta$.
A naive protocol is for Alice to send the indices of samples in $x_0+[-\Delta,\Delta]^d$. 
Locating each sample in that neighborhood requires on average $\Theta(\log\frac1{\Delta})=\Theta(\log k)$ bits.
Thus $\Theta(k/\log k)$ samples in that neighborhood can be located.
It turns out that the naive protocol is asymptotically optimal.

The proof of the lower bound (Section~\ref{sec_1pl}) follows by a reduction to testing independence for biased Bernoulli distributions, via a simulation argument.
Although some arguments are similar to \cite{hadar2019communication},
the present problem concerns \emph{biased} Bernoulli distributions instead.
The (KL) strong data processing constant turns out to be drastically different from the $\chi^2$ data processing constant, as opposed to the cases of many familiar distributions such as the unbiased Bernoulli or the Gaussian distributions.

As alluded, our main result is that the risk can be strictly improved when interactions are allowed:
\begin{thm}\label{thm4}
In interactive NP estimation, for any $\beta\in(0,\infty)$, $L\in(0,\infty)$, and $A\in[0,1)$, we have
\begin{align}
R(k)=\Theta\left(k^{-\frac{2\beta}{d+2\beta}}\right)
\label{e_rpb4} 
\end{align} 
where $\Theta(\cdot)$ hides multiplicative factors depending on $L$, $\beta$ and $A$.
\end{thm}
To achieve the scaling in \eqref{e_rpb4}, $r$ can grow 
as slowly as the super-logarithm (i.e., inverse tetration) of $k$; for the precise relation between $r$ and $k$, see
Section~\ref{sec_rinfty}.

The proof of the upper bound of Theorem~\ref{thm4} is given in 
Section~\ref{sec_d_interactive},
which is based on a novel multi-round estimation scheme for biased Bernoulli distributions formulated and analyzed in
Sections~\ref{sec_ipu},\ref{sec_exchanges},\ref{app4}.
Roughly speaking, the intuition is to ``locate'' the samples within neighborhoods of $(x_0,y_0)$ by successive refinements, which is more communication-efficient than revealing the location at once.

The lower bound of Theorem~\ref{thm4} is proved in Section~\ref{app9}.
The main technical hurdle is to develop new and tighter bounds on the symmetric data processing constant in \cite{liu2017secret} for the biased binary cases.

\section{Estimation of Biased Bernoulli Distributions}
\label{sec_ipu}
In this section, we shall describe an algorithm for estimating the joint distribution of a pair of biased Bernoulli random variables.
The biased Bernoulli estimation problem can be viewed as a  natural generalization of the Gap hamming problem \cite{Indyk}\cite{cr12} to the sparse setting,
and is the key component in both the upper and lower bound analysis for the nonparametric estimation problem.
Indeed, we shall explain in Section~\ref{app2} that our nonparametric estimator is based on a linear combination of rectangle kernel estimators, 
which estimate the probability that $X$ and $Y$ fall into  neighborhoods of $x_0$ and $y_0$.
Indicators that samples are within such neighborhoods are Bernoulli variables, so that the biased Bernoulli estimator can be used.
For the lower bound, we shall explain in Section~\ref{sec_1pl} that the nonparametric estimation problem can be reduced to the biased Bernoulli estimation problem via a simulation argument. 

For notational simplicity, we shall use $X,Y$ for the Bernoulli variables in this section as well as Sections~\ref{sec_exchanges}-\ref{app4}, although we should keep in mind that these are not the continuous variables in the original nonparametric estimation problem.

\noindent
{\bf Bernoulli Estimation Problem:}
\begin{itemize}
\item Fixed real numbers $m_1,m_2\in(10,\infty)$, and an unknown $\delta\in[-1,\min\{m_1,m_2\}-1]$.
\item ${\bf (X,Y)}=(X(l),Y(l))_{i=1}^{\infty}$ i.i.d.\ according to the distribution
\begin{align}
&P^{(\delta)}_{XY}:=
\nonumber
\\
&\left(
\begin{array}{cc}
\frac1{m_1m_2}(1+\delta)  & \frac1{m_1}(1-\frac1{m_2})-\frac{\delta}{m_1m_2}      \\
\frac1{m_2}(1-\frac1{m_1})-\frac{\delta}{m_1m_2} 
&     (1-\frac1{m_1})(1-\frac1{m_2})+\frac{\delta}{m_1m_2} 
\end{array}
\right)
\label{e_pdelta}
\end{align}
where we recall our convention that the upper left entry of the matrix denotes the probability that $X=Y=0$.
Alice observes $(X(l))_{l=1}^{\infty}$ and Bob observes $(Y(l))_{l=1}^{\infty}$.
\item Unlimited common randomness $\Pi_0$. 
\end{itemize}
{\bf Goal:} Alice and Bob exchange messages in no more than $r$ rounds in order to estimate $\delta$.

Our algorithm is described as follows:

\noindent
{\bf Input.} $m_1,m_2\in(10,\infty)$; positive integer $n$ and $r$; a sequence of real numbers $\alpha_1,\dots,\alpha_r\in(1,\infty)$ satisfying
\begin{align}
\prod_{1\le i\le r}^{\rm odd}\alpha_i\le \frac{m_1}{10};
\label{e_alpha1}
\\
\prod_{1\le i\le r}^{\rm even}\alpha_i\le \frac{m_2}{10}.
\label{e_alpha2}
\end{align}
The $\alpha_1,\dots,\alpha_r$ can be viewed as parameters of the algorithm,
and controls how much information is revealed about the locations of ``common zeros'' of ${\bf X,Y}$ in each round of communication. 
For example, setting $\alpha_1=\frac{m_1}{10}$ and all other $\alpha_i=1$ yields a one-way communication protocol, 
whereas setting all $\alpha_i>1$ yields a ``successive refinement'' algorithm which may incur smaller communication budget yet convey the same amount of information.

Before describing the algorithm, let us define a conditional distribution $P_{U^r|XY}$ by recursion,
which will be used later in generation random codebooks.
\begin{defn}\label{defn_cond}
For each $i\in\{1,\dots,r\}\setminus 2\mathbb{Z}$,
define
\begin{align}
P_{U_i|X=0,U^{i-1}={\bf0}}&=[1,0];
\label{e40}
\\
P_{U_i|X=1,U^{i-1}={\bf0}}&=[\alpha_i^{-1},1-\alpha_i^{-1}];
\\
P_{U_i|X=0,U^{i-1}\neq {\bf0}}&=P_{U_i|X=1,U^{i-1}\neq {\bf0}}=[0,1].
\label{e42}
\end{align}
Then set $P_{U_i|XYU^{i-1}}=P_{U_i|XU^{i-1}}$.
For $i=1,\dots,r$ even, we use similar definitions, but with the roles of $X$ and $Y$ switched.
This specifies $P_{U_i|XYU^{i-1}}$, $i=1,\dots,r$.
\end{defn}
Note that by Definition~\ref{defn_cond}, $U_i=1$ implies $U_{i+1}=1$ for each $i=1,\dots,r-1$.
In words, for $i$ odd,
$U_i$ marks all $X=0$ as $0$,
and marks $X=1$ as either $0$ or $1$;
whenever $U_i=1$ is marked, then $X$ is definitely 1, and will be forgotten in all subsequent rounds. 
Now set
\begin{align}
P_{XYU^r}^{(\delta)}:=P_{U^r|XY}P^{(\delta)}_{XY}.
\label{e_joint}
\end{align}
where $P_{U^r|XY}$ is induced by $(P_{U_i|XYU^{i-1}})_{i=1}^r$ in Definition~\ref{defn_cond}.

\noindent
{\bf Initialization. }
By applying a common function to the common randomness,
Alice and Bob can produce a shared infinite array $(V_{i,j}(l))$, where $i\in\{1,\dots,r\}$, $j\in\{1,2,\dots\}$, $l\in\{1,2,\dots,n\}$, 
such that the entries in the array are independent random variables, with 
$V_{i,j}(l)\sim {\rm Bern}(1-\alpha_i^{-1})$.
Also set
\begin{align}
U_0(l)=0,\quad\forall l=1,\dots,n.
\end{align}

\noindent
{\bf Iterations. }
Consider any $i=1,\dots,r$, where $i$ is odd. 
We want to generate ${\bf U}_i$ by selecting a codeword so that $({\bf X,Y},{\bf U}^i)$
follows the distribution of $(P_{XYU^i}^{(\delta)})^{\otimes n}$, where ${\bf U}^{i-1}$ is defined in previous rounds.
Define
\begin{align}
\mathcal{A}_0&:=\{l\le n\colon X(l)=0,U_{i-1}(l)=0\};
\\
\mathcal{A}_1&:=\{l\le n\colon X(l)=1,U_{i-1}(l)=0\};
\\
\mathcal{A}&:=\{l\le n\colon U_{i-1}(l)=0\}.
\end{align}
Note that Alice knows both $\mathcal{A}_0$ and $\mathcal{A}_1$, while Bob knows $\mathcal{A}$, since it will be seen from the recursion that Alice and Bob both know ${\bf U}_1,\dots,{\bf U}_{i-1}$ at the beginning of the $i$-th round.
Alice chooses $\hat{j}_i$ as the minimum nonnegative integer $j$ such that 
\begin{align}
V_{i,j}(l)=0,\quad\forall l\in\mathcal{A}_0.
\label{e125}
\end{align}
Alice encodes $\hat{j}_i$ using a prefix code, e.g.\,
Elias gamma code \cite{1055349},
 and sends it to Bob.
Then both Alice and Bob compute ${\bf U}_i=(U_i(l))_{l=1}^n
\in\{0,1\}^n$ by 
\begin{align}
U_i(l)&:=V_{(i,\hat{j}_i)}(l),\quad\forall l\in\mathcal{A};
\label{e30}
\\
U_i(l)&:=1,\quad\forall l\in\{1,\dots,n\}\setminus\mathcal{A}.
\label{e31}
\end{align}
The operations in the $i$-th round for even $i$ is similar, with the roles of Alice and Bob reversed.
We will see later that the notation ${\bf U}_i$ is consistent in the sense of \eqref{e_consistent}.

\noindent
{\bf Estimator.}
Recall that in classical parametric statistics, 
one can evaluate the score function at the sample,
compute its expectation and variation, 
and construct an estimator achieving the 
Cramer-Rao bound asymptotically.
Now for $i\in\{1,\dots,r\}\setminus 2\mathbb{Z}$,
define the score function
\begin{align}
\Gamma_i(u^i,y)
&:=\frac{\partial}{\partial\delta}\left.\ln P^{(\delta)}_{U_i|YU^{i-1}}(u_i|y,u^{i-1})\right|_{\delta=0}
\label{e_gamma}
\\
&=\left\{
\begin{array}{ccc}
\frac{\partial}{\partial\delta}\left.\ln P^{(\delta)}_{U_i|YU^{i-1}}(u_i|{y,\bf 0})\right|_{\delta=0}  &   \textrm{if }u^{i-1}={\bf 0}  \\
0 &    \textrm{otherwise}
\end{array}
\right.
\label{e_gammab}
\end{align}
where $P^{(\delta)}_{U_i|YU^{i-1}}$ is induced by $P_{XYU^r}^{(\delta)}$.
For $i\in\{1,\dots,r\}\cup 2\mathbb{Z}$,
define $\Gamma_i(u^i,x)$ similarly with the roles of $X$ and $Y$ reversed.
Alice and Bob can each compute
\begin{align}
\Gamma^{\rm A}:=\sum_{1\le i\le r}^{\rm even}\sum_{l=1}^n\Gamma_i(U^i(l),X(l))
\end{align}
and
\begin{align}
\Gamma^{\rm B}:=\sum_{1\le i\le r}^{\rm odd}\sum_{l=1}^n\Gamma_i(U^i(l),Y(l))
\end{align}
respectively.
Finally, Alice's and Bob's estimators are given by
\begin{align}
\hat{\delta}^{\rm A}&:=\Gamma^{\rm A}\cdot\left(\partial_{\delta}\mathbb{E}^{(\delta)}[\Gamma^{\rm A}]\right)^{-1};
\\
\hat{\delta}^{\rm B}&:=\Gamma^{\rm B}\cdot\left(\partial_{\delta}\mathbb{E}^{(\delta)}[\Gamma^{\rm B}]\right)^{-1},
\label{e_deltab}
\end{align}
where $\mathbb{E}^{(\delta)}$ refers to expectation when the true parameter is $\delta$, and $\partial_{\delta}$ denotes the derivative in $\delta$.
We will show that these estimators are well-defined: 
$\partial_{\delta}\mathbb{E}^{(\delta)}[\Gamma^{\rm A}]$ and $\partial_{\delta}\mathbb{E}^{(\delta)}[\Gamma^{\rm B}]$ are independent of $\delta$ (Lemma~\ref{lem_linear}), and can be computed by Alice and Bob without knowing $\delta$.
\begin{lem}
For each $i\in\{1,\dots,r\}\setminus2\mathbb{Z}$, conditioned on ${\bf X, Y},{\bf U}_1,\dots, {\bf U}_{i-1}$, we have 
\begin{align}
{\bf U}_i\sim P^{\otimes n}_{U_i|XU^{i-1}}(\cdot|{\bf X},{\bf U}_1,\dots,{\bf U}_{i-1}),
\label{e_consistent}
\end{align}
where $P_{U_i|XU^{i-1}}$ is as defined in
\eqref{e40}-\eqref{e42}.
A similar relation holds for even $i$.
\end{lem}
\begin{proof}
Immediate from \eqref{e30}-\eqref{e31}.
\end{proof}
\begin{lem}\label{lem_linear}
$\mathbb{E}^{(\delta)}[\Gamma^{\rm A}]$ and $\mathbb{E}^{(\delta)}[\Gamma^{\rm B}]$ are linear in $\delta$.
\end{lem}
\begin{proof}
By \eqref{e_consistent},
\begin{align}
&\quad\mathbb{E}^{(\delta)}[\Gamma_i(U^i(l)),Y(l)]
\\
&=\sum_{u^r,x,y}\Gamma_i(u^i,y)P^{(\delta)}_{XY}(x,y)P_{U^r|XY}(u^r|x,y)
\label{e_48}
\end{align}
for each $i$ odd and $l\in\{1,\dots,n\}$, and similar expressions hold for $i$ even.
The claims then follow.
\end{proof}
\begin{thm}\label{thm6}
$\hat{\delta}^{\rm A}$ and $\hat{\delta}^{\rm B}$ are unbiased estimators.
\end{thm}
\begin{proof}
For $i$ odd, by \eqref{e_gamma} we have
$\sum_{u_i}\Gamma_i(u^i,y)P^{(0)}_{U_i|YU^{i-1}}(u_i|y,u^{i-1})=0$ for any $(y,u^{i-1})$.
Then $\mathbb{E}^{(0)}[\Gamma_i(U^i(l)),Y(l)]=0$ follows from \eqref{e_48}.
It follows that $\mathbb{E}^{(0)}[\hat{\delta}^{\rm A}]=\mathbb{E}^{(0)}[\hat{\delta}^{\rm B}]=0$, and unbiasedness is implied by Lemma~\ref{lem_linear}.
\end{proof}

\section{Bounds on Information Exchanges}\label{sec_exchanges}
In this section we prove key auxiliary results that will be used in the upper bounds.

\subsection{General $(\alpha_i)$}
The following Theorem is crucial for the achievability part of the analysis of the Bernoulli estimation problem described in Section~\ref{sec_ipu} (and hence for the nonparametric estimation problem).
Specifically, \eqref{e43}-\eqref{e44} bounds the communication from Alice to Bob and in reverse, and \eqref{e47}-\eqref{e48} bounds the information exchanged which, in turn, will bound the estimation risk via Fano's inequality.
\begin{thm}\label{thm_s}
Consider any $m_1,m_2>10$,
$\alpha_1,\dots,\alpha_r\in(1,\infty)$ satisfying \eqref{e_alpha1}-\eqref{e_alpha2},
and $P_{U^rXY}^{(\delta)}$ as in \eqref{e_joint}.
We have
\begin{align}
\sum_{1\le i\le r}^{\rm odd}P^{(0)}_{XU^{i-1}}(0,{\bf 0})\log\alpha_i
&\le 
\frac{1.1}{m_1} \sum_{1\le i\le r}^{\rm odd}\log\alpha_i\prod_{2\le j\le i-1}^{\rm even}\alpha_j^{-1};
\label{e43}
\\
\sum_{1\le i\le r}^{\rm even}P^{(0)}_{YU^{i-1}}(0,{\bf 0})\log\alpha_i
&\le
\frac{1.1}{m_2} 
\sum_{1\le i\le r}^{\rm even}\log\alpha_i\prod_{1\le j\le i-1}^{\rm odd}\alpha_j^{-1},
\label{e44}
\end{align}
and assuming the natural base of logarithms,
\begin{align}
\lim_{\delta\to0}\delta^{-2}\sum_{1\le i\le r}^{\rm odd}I(U_i;Y|U^{i-1})
&\ge
\frac1{5m_1^2m_2}
\prod_{1\le j\le r}^{\rm odd}\alpha_j;
\label{e47}
\\
\lim_{\delta\to0}\delta^{-2}
\sum_{1\le i\le r}^{\rm even}I(U_i;X|U^{i-1})
&\ge
\frac1{5m_1m_2^2}
\prod_{1\le j\le r}^{\rm even}\alpha_j.
\label{e48}
\end{align}
\end{thm}
The proof can be found in Appendix~\ref{app_thm_s}.

\begin{rem}\label{rem3}
Since
\begin{align}
&\quad P^{(0)}_{XU^{i-1}}(0,{\bf 0})\log\alpha_i
\nonumber\\
&=
P^{(0)}_{XU^{i-1}}(0,{\bf 0})D(P_{U_i|X=0,U^{i-1}={\bf0}}\|P_{U_i|X=1,U^{i-1}={\bf0}})
\\
&\ge P^{(0)}_{U^{i-1}}({\bf 0})\inf_Q\left[
P^{(0)}_{X|U^{i-1}}(0|{\bf 0})
D(P_{U_i|X=0,U^{i-1}={\bf0}}\|Q)
\right.
\nonumber\\
&\quad+\left.
P^{(0)}_{X|U^{i-1}}(1|{\bf 0})
D(P_{U_i|X=1,U^{i-1}={\bf0}}\|Q)
\right]
\\
&=P^{(0)}_{U^{i-1}}({\bf 0})I(U_i;X|U^{i-1}={\bf0})
\\
&=I(U_i;X|U^{i-1}),
\end{align}
Theorem~\ref{thm_s} also implies the following bound on the external information
(see \eqref{e_extrinsic}):
\begin{align}
I(U^r;XY)&\le 
\frac{1.1}{m_1} 
\sum_{1\le i\le r}^{\rm odd}\log\alpha_i\prod_{2\le j\le i-1}^{\rm even}\alpha_j^{-1}
\nonumber\\
&\quad+
\frac{1.1}{m_2} 
\sum_{1\le i\le r}^{\rm even}\log\alpha_i\prod_{1\le j\le i-1}^{\rm odd}\alpha_j^{-1}.
\end{align}
\end{rem}

\begin{rem}
Let us provide some intuition why interaction helps, assuming the case of $m_1=m_2=m$ for simplicity.
From the proof of Theorem~\ref{thm_s}, it can be seen that up to a constant factor, $s_{\infty}^*(X;Y)$ is at least $\frac{\delta^2}{m^3}\int_0^{\ln\frac{m}{100}}e^t{\rm d}t\left(\frac1{m}\int_0^{\ln\frac{m}{100}}e^{-t}{\rm d}t\right)^{-1}\sim\frac{\delta^2}{m}$ 
(which is in fact sharp as will be seen from the upper bound on $s_{\infty}^*(X;Y)$ in  Theorem~\ref{thm7}).
Moreover, lower bounds on $s_r^*(X;Y)$ can be computed by replacing the integrals with discrete sums with $r$ terms: 
\begin{align}
\frac{\delta^2}{m^3}\sum_{i=1}^{\lceil r/2\rceil}(e^{t_i}-e^{t_{i-1}})\left(\frac1{m}\sum_{i=1}^{\lceil r/2\rceil}e^{-t_{i-1}}(t_i-t_{i-1})\right)^{-1}
\end{align}
where $1=t_0<t_1<\dots<t_{\lceil r/2\rceil}=\ln\frac{m}{100}$.
In particular, when $r=1$, we recover $s_1^*(X;Y)\sim\frac{\delta^2}{m\ln m}$, whereas choosing $t_i-t_{i-1}=1$, $i=1,\dots,\lceil r/2\rceil$ shows that $r\sim \ln m$ achieves $s_r^*(X;Y)\sim \frac{\delta^2}{m}$.
Even better, later we will take $t_k$ as the $k$-th iterated power of $2$, and then $r$ will be the super logarithm of $m$. 
\end{rem}
Recall that $(\alpha_i)$ control the amount of information revealed in each round and  serve as hyperparameters of the algorithm to be tuned.
Next we shall explain how to select the value of $(\alpha_i)$ so that the optimal performance is achieved in the one-way and interactive settings.

\subsection{$r=1$ Case}
Specializing Theorem~\ref{thm_s} we obtain:
\begin{cor}\label{cor_4}
For any $m_1,m_2>10$,
with $r=1$ and $\alpha_1=\frac{m_1}{10}$ we have
\begin{align}
P^{(0)}_X(0)\log\alpha_1
&\le \frac{1.1}{m_1}\log\frac{m_1}{10};
\\
\lim_{\delta\to0}\delta^{-2}I(U_1;Y)
&\ge \frac1{50m_1m_2}.
\end{align}
\end{cor}

\subsection{$r=\infty$ Case}\label{sec_rinfty}
Denote by ${}^n2$ the $n$-th tetration of 2, which is defined recursively by ${}^02=1$ and
\begin{align}
{}^n2:=2^{\left({}^{(n-1)}2\right)}, \quad \forall n\ge 1.
\end{align}
Let $m:=\min\{m_1,m_2\}$, and let $r_0$ be the minmum integer such that
\begin{align}
\exp_e({}^{r_0}2-1)\ge \frac{m}{10}.
\label{e33}
\end{align}
For $m>10$ we have $r_0\ge1$.
Then we set
\begin{align}
r&:=2r_0;
\label{e112}
\\
\alpha_{2k-1}&:=\alpha_{2k}:=\exp_e({}^k2-{}^{(k-1)}2),\quad\forall k\in\{1,\dots,r_0-1\};
\\
\alpha_{2r_0-1}&:=\alpha_{2r_0}=\frac{m}{10}\exp_e(1-{}^{(r_0-1)}2),
\label{e114}
\end{align}
which fulfills $\alpha_i>1$.
We see that
\begin{align}
&\quad\sum_{1\le i\le r}^{\rm odd}\ln\alpha_i\prod_{2\le j\le i-1}^{\rm even}\alpha_j^{-1}
\nonumber\\
&\le\sum_{k=1}^{r_0}\left({}^k2-{}^{(k-1)}2\right)\exp_e\left(1-{}^{(k-1)}2\right)
\label{e115}
\\
&\le e\sum_{k=1}^{\infty}{}^k2\exp_e\left(-{}^{(k-1)}2\right)
\\
&=e\sum_{k=1}^{\infty}\exp_e\left(-(1-\log2)\cdot{}^{(k-1)}2\right)
\\
&< 5.
\end{align}
The first inequality above follows by $\alpha_{r-1}=\frac{m}{10}\exp_e(1-{}^{(r_0-1)}2)\le \exp_e({}^{r_0}2-{}^{(r_0-1)}2)$.
Similarly we also have $\sum_{1\le i\le r}^{\rm even}\ln\alpha_i\prod_{1\le j\le i-1}^{\rm odd}\alpha_j^{-1}<5$.
Moreover,
\begin{align} 
\prod_{1\le j\le r}^{\rm odd}\alpha_j
=
\prod_{1\le j\le r}^{\rm even}\alpha_j
=\frac{m}{10}.
\end{align}
Summarizing, we have

\begin{cor}\label{cor4}
Consider $m_1,m_2>10$, $m:=\min\{m_1,m_2\}$, and $(\alpha_i)$ defined in \eqref{e112}-\eqref{e114}.
We have
\begin{align}
\sum_{1\le i\le r}^{\rm odd}P^{(0)}_{XU^{i-1}}(0,{\bf 0})\ln\alpha_i
&\le \frac{6}{m_1};
\\
\sum_{1\le i\le r}^{\rm even}P^{(0)}_{YU^{i-1}}(0,{\bf 0})\ln\alpha_i
&\le \frac{6}{m_2};
\\
\lim_{\delta\to0}\delta^{-2}\sum_{1\le i\le r}^{\rm odd}I(U_i;Y|U^{i-1})
&\ge \frac{m}{50m_1^2m_2};
\\
\lim_{\delta\to0}\delta^{-2}\sum_{1\le i\le r}^{\rm even}I(U_i;X|U^{i-1})
&\ge \frac{m}{50m_1m_2^2}.
\end{align}
where $r=2r_0$ and $r_0$ is defined in \eqref{e33}.
\end{cor}
Let us remark that the sequence $(\alpha_i)$ we used in \eqref{e112}-\eqref{e114} is essentially optimal:
Let $\beta_k:=\prod_{2\le j\le 2k}^{\rm even}\alpha_j^{-1}$.
In order for \eqref{e115} to converge, we need $\sum_k\ln(\frac{\beta_k}{\beta_{k-1}})\beta_{k-1}^{-1}$ to be convergent.
Therefore $\beta_k$ cannot grow faster than $\beta_k=\exp(\beta_{k-1})$ which is tetration.
However this only amounts to a lower bound on $r$ for a particular design of $P_{U^r|XY}$ in Definition~\ref{defn_cond}.
Since tetration grows super fast, from a practical viewpoint $r$ is essentially a constant.
Nevertheless, it remains an interesting theoretical question whether $r$ needs to diverge:
\begin{conj}\label{conj1}
If there is an algorithm (indexed by $k$) achieving the optimal risk \eqref{e_rpb4} for nonparametric estimation, then necessarily $r\to\infty$ as $k\to\infty$.
\end{conj}


\section{Achievability Bounds for Bernoulli Estimation}\label{app4}
In this section we analyze the performance of the Bernoulli distribution estimation algorithm described in Section~\ref{sec_ipu}.
\subsection{Communication Complexity}
Consider any $i\in\{1,\dots,r\}$. 
Denoting by $\widehat{P}_{{\bf XY}{\bf U}^i}$ the empirical distribution of $(X(l),Y(l),U_1(l),\dots,U_i(l))_{l=1}^n$, we have from \eqref{e_consistent} that
\begin{align}
\mathbb{E}^{(\delta)}[\widehat{P}_{{\bf XY}{\bf U}^i}|{\bf X,Y,}{\bf U}^{i-1}]
=\widehat{P}_{{\bf XY}{\bf U}^{i-1}}P_{U_i|XU^{i-1}}.
\label{e116}
\end{align}
In particular,
\begin{align}
\mathbb{E}^{(\delta)}[\widehat{P}_{{\bf XY}{\bf U}^r}]=P_{XY}^{(\delta)}P_{U^r|XY}.
\label{e_125}
\end{align}
Let $\ell(\hat{j}_i):=2\lfloor\log_2(\hat{j}_i)\rfloor+1$ be the number of bits need to encode the positive integer $\hat{j}_i$ using the Elias gamma code \cite{1055349}.
For each $i\in\{1,\dots,r\}\cap2\mathbb{Z}$ we have
\begin{align}
\mathbb{E}^{(\delta)}[\ell(\hat{j}_i)|{\bf X,Y,}{\bf U}^{i-1}]
&\le 2\mathbb{E}^{(\delta)}[\log_2\hat{j}_i|{\bf X,Y,}{\bf U}^{i-1}]+1
\\
&\le 2\log_2\mathbb{E}^{(\delta)}[\hat{j}_i|{\bf X,Y,}{\bf U}^{i-1}]+1
\\
&=2\log_2\alpha_i^{n\widehat{P}_{{\bf X U}_{i-1}}(0,{\bf 0})}+1
\label{e128}
\\
&=2n\widehat{P}_{{\bf X U}_{i-1}}(0,{\bf 0})\log_2\alpha_i+1
\end{align}
where \eqref{e128} follows from the selection rule \eqref{e125} and the formula of expectation of the geometric distribution.
Then
\begin{align}
\mathbb{E}^{(\delta)}[\ell(\hat{j}_i)]
&\le 2nP_{XU_{i-1}}^{(\delta)}(0,{\bf 0})\log_2\alpha_i+1
\label{e_136}
\\
&\le 2(1+\delta)nP_{XU_{i-1}}^{(0)}(0,{\bf 0})\log_2\alpha_i+1
\label{e122}
\end{align}
where \eqref{e_136} used \eqref{e_125}; \eqref{e122} used the fact that $P_{XYU^r}^{(\delta)}$ is dominated by $(1+\delta)P_{XYU^r}^{(0)}$.
Note that ${\bf 0}$ in \eqref{e_136}-\eqref{e122} denotes the value of the vector $U^{i-1}$.

\subsection{Expectation of $\Gamma^{\rm B}$}
Recall that $\Gamma^{\rm B}$ was defined in \eqref{e_gammab}.
Pick arbitrary $i\in\{1,\dots,r\}\setminus2\mathbb{Z}$.
Since
\begin{align}
P_{U^iY}^{(\delta)}:=P_Y\prod_{j=1}^i P^{(\delta)}_{U_j|U^{j-1}Y}
\end{align}
 and since $P_Y$ and $(P^{(\delta)}_{U_j|U^{j-1}Y})_{j\in\{1,\dots,r\}\cap 2\mathbb{Z}}$ are independent of $\delta$, we obtain
\begin{align}
\partial_{\delta}\ln P_{U^iY}^{(\delta)}(u^i,y)|_{\delta=0}
=\sum_{1\le j\le i}^{\rm odd}\Gamma_j(u^j,y).
\label{e132}
\end{align}
Next, observe that for any $l\in\{1,\dots,n\}$,
\begin{align}
&\quad\mathbb{E}^{(0)}[\Gamma_i(U^i(l),Y(l))|U^{i-1}(l),Y(l)]
\nonumber\\
&=\mathbb{E}^{(0)}\left[\left.
\frac{\left.\partial_{\delta}P^{(\delta)}_{U_i|YU^{i-1}}(U_i(l)|Y,U^{i-1}(l))
\right|_{\delta=0}}
{P^{(0)}_{U_i|YU^{i-1}}(U_i(l)|Y,U^{i-1}(l))}
\right|U^{i-1}(l),Y(l)
\right]
\\
&=\sum_{u_i}\left.\partial_{\delta}P^{(\delta)}_{U_i|YU^{i-1}}(u_i|Y,U^{i-1}(l))
\right|_{\delta=0}
\\
&=0.
\label{e145}
\end{align}
Moreover, for any $\delta\neq0$,
\begin{align}
&\quad\frac1{\delta}\sum_{l=1}^n\mathbb{E}^{(\delta)}[\Gamma_i(U_1(l),\dots,U_i(l),Y(l))]
\nonumber\\
&=\delta^{-1}n\sum_{u^i,y}\Gamma_i(u^i,y)P^{(\delta)}_{U^iY}(u^i,y)
\label{e136}
\\
&=n\sum_{u^i,y}\Gamma_i(u^i,y)\frac{\partial}{\partial\delta}P_{U^iY}^{(\delta)}(u^i,y)|_{\delta=0}
\label{e137}
\\
&=n\sum_{u^i,y}\Gamma_i(u^i,y)P_{U^iY}^{(0)}(u^i,y)\sum_{1\le j\le i}^{\rm odd}\Gamma_j(u^j,y)
\label{e139}
\\
&=n\sum_{u^i,y}\Gamma_i^2(u^i,y)P_{U^iY}^{(0)}(u^i,y)
\label{e140}
\end{align}
where \eqref{e136} used \eqref{e_125};
\eqref{e137} used \eqref{e145} and the linearity of $P^{\delta}_{U^{i-1}Y}$ in $\delta$;
\eqref{e139} used \eqref{e132};
\eqref{e140} follows from \eqref{e145}.
Thus
\begin{align}
\frac1{\delta}\mathbb{E}^{(\delta)}[\Gamma^{\rm B}]=I^{\rm B},
\quad\forall \delta\neq 0
\label{e141}
\end{align}
where we defined
\begin{align}
I^{\rm B}&:=
n\sum_{1\le i\le r}^{\rm odd}\mathbb{E}^{(0)}[\Gamma_i^2(U^i(1),Y(1))].
\end{align}

\begin{lem}\label{lem5}
Let $(U^i,Y)\sim P^{(\delta)}_{U^iY}$.
We have
\begin{align}
I^{\rm B}\ge 2
n\lim_{\delta\to0}\delta^{-2}\sum_{1\le i\le r}^{\rm odd}I(U_i;Y|U^{i-1})
\end{align}
where the logarithmic base of the mutual information is natural.
\end{lem}
\begin{proof}
Consider any $i\in\{1,2,\dots,r\}\setminus 2\mathbb{Z}$. 
we have
\begin{align}
&\quad\mathbb{E}^{(0)}\left[\Gamma_i^2(U^i,Y)\right]
\nonumber\\
&=\mathbb{E}^{(0)}\left[\left(\partial_{\delta}\ln P^{(\delta)}_{U_i|YU^{i-1}}(U_i|YU^{i-1})|_{\delta=0}\right)^2\right]
\\
&=2\lim_{\delta\to 0}\delta^{-2}D(P_{U_i|YU^{i-1}}^{(\delta)}\|P_{U_i|YU^{i-1}}^{(0)}|P_{U^{i-1}Y}^{(0)})
\\
&=2\lim_{\delta\to 0}\delta^{-2}D(P_{U_i|YU^{i-1}}^{(\delta)}\|P_{U_i|U^{i-1}}^{(0)}|P_{U^{i-1}Y}^{(0)})
\label{e152}
\\
&=2\lim_{\delta\to 0}\delta^{-2}D(P_{U_i|YU^{i-1}}^{(\delta)}\|P_{U_i|U^{i-1}}^{(0)}|P_{U^{i-1}Y}^{(\delta)})
\label{e147}
\\
&\ge 2\lim_{\delta\to 0}\delta^{-2}D(P_{U_i|YU^{i-1}}^{(\delta)}\|P_{U_i|U^{i-1}}^{(\delta)}|P_{U^{i-1}Y}^{(\delta)})
\\
&=\lim_{\delta\to0}\delta^{-2}I(U_i;Y;U^{i-1})
\end{align}
where we defined the conditional KL divergence $$D(P_{Y|X}\|Q_{Y|X}|P_X):=\int D(P_{Y|X=x}\|Q_{Y|X=x})dP_X(x);$$
\eqref{e152} follows since $P^{(0)}_{U_i|YU^{i-1}}=P^{(0)}_{U_i|U^{i-1}}$;
\eqref{e147} follows since $\lim_{\delta\to0}P^{(\delta)}_{U^{i-1}Y}=P^{(0)}_{U^{i-1}Y}$.
\end{proof}

\subsection{Variance of $\Gamma^{\rm B}$}
For any $\delta$, since $({\bf U}^r,{\bf X,Y})\sim (P^{(\delta)}_{U^rXY})^{\otimes n}$, we have
\begin{align}
\var^{(\delta)}(\Gamma^{\rm B})
&=\sum_{l=1}^n\var^{(\delta)}\left(\sum_{1\le i\le r}^{\rm odd}\Gamma_i(U_1(l),\dots,U_i(l),Y(l))\right)
\\
&=n\var^{(\delta)}\left(\sum_{1\le i\le r}^{\rm odd}\Gamma_i(U_1(1),\dots,U_i(1),Y(1))\right).
\end{align}
However,
\begin{align}
&\quad\var^{(\delta)}\left(\sum_{1\le i\le r}^{\rm odd}\Gamma_i(U_1(1),\dots,U_i(1),Y(1))\right)
\nonumber\\
&\le
\mathbb{E}^{(\delta)}\left[\left(\sum_{1\le i\le r}^{\rm odd}\Gamma_i(U_1(1),\dots,U_i(1),Y(1))\right)^2\right]
\\
&\le (1+\delta)\mathbb{E}^{(0)}\left[\left(\sum_{1\le i\le r}^{\rm odd}\Gamma_i(U_1(1),\dots,U_i(1),Y(1))\right)^2\right]
\label{e154}
\\
&\le (1+\delta)\sum_{1\le i\le r}^{\rm odd}\mathbb{E}^{(0)}\left[\Gamma_i^2(U_1(1),\dots,U_i(1),Y(1))\right]
\label{e155}
\end{align}
where \eqref{e154} follows since $P^{(\delta)}$ is dominated by $(1+\delta)P^{(0)}$;
\eqref{e155} used \eqref{e145}.
Therefore
\begin{align}
&\quad\var^{(\delta)}(\Gamma^{\rm B})
\nonumber\\
&\le n(1+\delta)\sum_{1\le i\le r}^{\rm odd}\mathbb{E}^{(0)}\left[\Gamma_i^2(U_1(1),\dots,U_i(1),Y(1))\right]
\\
&=(1+\delta)I^{\rm B}.
\label{e156}
\end{align}

\subsection{$r=1$ Case}
We now prove achievability bounds for the Bernoulli distribution estimation algorithm.
\begin{cor}\label{cor7}
Given $m_1,m_2>10$, for $r=1$ and $\alpha_1:=\frac{m_1}{10}$, the mean square error $\mathbb{E}[|\hat{\delta}^{\rm B}-\delta|^2]\le\frac{25(1+\delta)m_1m_2}{n}$ and total communication cost $\mathbb{E}[|\Pi_1|]\le \frac{2.2(1+\delta)n}{m_1}\log_2\frac{m_1}{10} +1$.
\end{cor}
\begin{proof}
We have
\begin{align}
\mathbb{E}[|\hat{\delta}^{\rm B}-\delta|^2]
&= \var^{(\delta)}(\Gamma^{\rm B})\cdot (\partial_{\delta}\mathbb{E}^{(\delta)}[\Gamma^{\rm B}])^{-2}
\label{e157}
\\
&\le (1+\delta)I^{\rm B}\cdot (I^{\rm B})^{-2}
\label{e158}
\\
&\le(1+\delta)\left[2
n\lim_{\delta\to0}\delta^{-2}I(U_1;Y)\right]^{-1}
\label{e159}
\\
&\le \frac{25(1+\delta)m_1m_2}{n}
\end{align}
where \eqref{e157} follows since $\hat{\delta}^{\rm B}$ is unbiased (Theorem~\ref{thm6});
\eqref{e158} follows from 
\eqref{e141} and
\eqref{e156};
\eqref{e159} follows from 
Lemma~\ref{lem5};
lastly we used Corollary~\ref{cor_4}.

As for the communication cost
\begin{align}
\mathbb{E}[|\Pi^{\rm A\to B}|]
&\le 2(1+\delta)n\sum_{1\le i\le 1}P_{XU_{i-1}}^{(0)}(0,0)\log_2\alpha_i+1
\\
&\le
\frac{2.2(1+\delta)n}{m_1}\log_2\frac{m_1}{10} +1
\end{align}
where we used \eqref{e122} and Corollary~\ref{cor_4}.
\end{proof}

\subsection{$r=\infty$ Case}
\begin{cor}\label{cor8}
Let $m_1,m_2>10$, $m:=\min\{m_1,m_2\}$.
For $r$, $(\alpha_i)$ defined in Section~\ref{sec_rinfty},
the mean square error $\mathbb{E}[|\hat{\delta}^{\rm B}-\delta|^2]\le \frac{25(1+\delta)m_1m_2^2}{nm}$ and total communication cost $\mathbb{E}[|\Pi^r|]\le 6(1+\delta)n(m_1^{-1}+m_2^{-1})\log_2e +\frac{r+1}{2}$.
\end{cor}
\begin{proof}
The bound on the mean square error is similar to the $r=1$ case:
\begin{align}
\mathbb{E}[|\hat{\delta}^{\rm B}-\delta|^2]
&\le \var^{(\delta)}(\Gamma^{\rm B})\cdot (\partial_{\delta}\mathbb{E}^{(\delta)}[\Gamma^{\rm B}])^{-2}
\\
&\le (1+\delta)I^{\rm B}\cdot (I^{\rm B})^{-2}
\\
&\le(1+\delta)\left[2
n\lim_{\delta\to0}\delta^{-2}\sum_{1\le i\le r}^{\rm odd}I(U_i;Y|U^{i-1})\right]^{-1}
\\
&\le \frac{25(1+\delta)m_1^2m_2}{mn}
\end{align}
except that we use Corollary~\ref{cor4} in the last step.

For the communication cost,
\begin{align}
\mathbb{E}[|\Pi^{\rm A\to B}|]
&\le 2(1+\delta)n\sum_{1\le i\le r}^{\rm odd}P_{XU_{i-1}}^{(0)}(0,0)\log_2\alpha_i+\frac{r+1}{2}
\\
&\le
6(1+\delta)n(m_1^{-1}+m_2^{-1})\log_2e +\frac{r+1}{2}
\end{align}
where used \eqref{e122} and Corollary~\ref{cor4}.
\end{proof}

\section{Density Estimation Upper Bounds}\label{app2}
In this section we prove the upper bounds in Theorem~\ref{thm3} and Theorem~\ref{thm4}, by building nonparametric density estimators based on the Bernoulli distribution estimator.
For $\beta\in(0,2]$, the rectangular kernel is minimax optimal (Section~\ref{sec_npestimation}), 
so that the integral with the kernel can be directly estimated using the Bernoulli distribution estimator, which we explain in Section~\ref{sec_d1way} and \ref{sec_d_interactive}.
Extension to higher order kernels is possible using a linear combination of rectangular kernels, which is explained in Section~\ref{sec_betage}.

\subsection{Density Lower Bound Assumption}
First, we observe the following simple argument showing that it suffices to consider $A>0$.
Define
\begin{align}
B:=\sup_{x,y\in[0,1]^d}\sup_{p_{XY}}p_{XY}(x,y),
\label{e_B}
\end{align}
where the supremum is over all density $p_{XY}$ on $\mathbb{R}^{2d}$ satisfying $\|p_{XY}\|_{[0,1]^{2d},\beta}\le L$.
Clearly $B>1$ is finite and depends only on $\beta,L,d$.
\begin{lem}\label{lem7}
In either the one-way or the interactive setting, suppose that there exists an algorithm achieving $\max_{p_{XY}\in \mathcal{H}(\beta,L,A)}\mathbb{E}[|\hat{p}-p_{XY}(x_0,y_0)|^2]\le R$ for some $R>0$ and $A\in(0,1)$.
Then, there must be an algorithm achieving $\max_{p_{XY}\in \mathcal{H}(\beta,L,0)}\mathbb{E}[|\hat{p}-p_{XY}(x_0,y_0)|^2]\le \left(\frac{1+A}{1-A}\right)^2R$.
\end{lem}
\begin{proof}
Pick one $p_{XY}$ such that 
$\|p_{XY}\|_{[0,1]^{2d},\beta}\le L$
and $\inf_{x,y\in[0,1]^d}p_{XY}(x,y)\ge \frac{1+A}{2}>A$.
Since $\frac{1+A}{2}<1$, such $p_{XY}$ exists.
 Consider an arbitrary $q_{XY}\in \mathcal{H}(\beta,L,0)$,
and suppose infinite pairs $(X_1,Y_1),\dots$ i.i.d.\ according to $q_{XY}$ are available to Alice and Bob.
Using the common randomness, Alice and Bob can simulate i.i.d.\ pairs $(\tilde{X}_1,\tilde{Y}_1),\dots$ according to $\tilde{p}_{XY}:=\frac{2A}{1+A}p_{XY}+\frac{1-A}{1+A}q_{XY}$, by replacing each pair with probability $\frac{2A}{1+A}$ with a new pair drawn according to $p_{XY}$.
Clearly $\tilde{p}_{XY}\in\mathcal{H}(\beta,L,A)$, and by assumption, $\tilde{p}_{XY}(x_0,y_0)$ can be estimated with mean square risk $R$.
Since $p_{XY}$ is known, this implies that $q_{XY}(x_0,y_0)$ can be estimated with mean square risk $\left(\frac{1+A}{1-A}\right)^2R$.
\end{proof}
For the rest of the section, we will assume that there is a density lower bound $A>0$ and $p_{XY}\in\mathcal{H}(\beta,L,A)$, which is sufficient in view of Lemma~\ref{lem7}.
Consider bandwidth $h>0$ (which will be specified later as an inverse polynomial of $k$).
Also introduce the notations
\begin{align}
\mathcal{A}&:=x_0+h[-1,1]^d;
\\
\mathcal{B}&:=y_0+h[-1,1]^d.
\end{align}
Define $P_{\underline{X}\underline{Y}}$ as the probability distribution induced by $P_{XY}$ and with
\begin{align}
\underline{X}&:=1\{X\notin \mathcal{A}\};
\\
\underline{Y}&:=1\{Y\notin \mathcal{B}\}.
\end{align}
Define
\begin{align}
m_1&:=\mathbb{P}^{-1}(\underline{X}=0); 
\label{em1}
\\
m_2&:=\mathbb{P}^{-1}(\underline{Y}=0).
\label{em2}
\end{align}
Note that $m_1/m_2$ is bounded above and below by positive constants depending on $A$, $\beta$, and $L$ (see \eqref{e174} and \eqref{e178}).
Also, we can assume Alice and Bob both know $m_1$ and $m_2$, since with infinite samples Alice and Bob know their marginal densities $p_X$ and $p_Y$, and Alice can send $m_1$ to Bob with very high precision using negligible number of bits.
Let $\delta\ge -1$ be the number such that $P_{\underline{X}\underline{Y}}$ is the matrix 
\begin{align}
\left(
\begin{array}{cc}
\frac1{m_1m_2}(1+\delta)  & \frac1{m_1}(1-\frac1{m_2})-\frac{\delta}{m_1m_2}      \\
\frac1{m_2}(1-\frac1{m_1})-\frac{\delta}{m_1m_2} 
&     (1-\frac1{m_2})^2+\frac{\delta}{m_1m_2} 
\end{array}
\right).
\end{align}
Let $\hat{\delta}^{\rm B}$ be Bob's estimator  of $\delta$ in \eqref{e_deltab}.
Then we define Bob's density estimator:
\begin{align}
\hat{p}^{\rm B}:=\frac{1+\hat{\delta}^{\rm B}}{m_1m_2h^{2d}}.
\end{align}

We next show that the smoothness of the density ensures that $1+\delta$ is at most the order of a constant.
Recall that $A$ is a density lower bound on $p_X$ and $p_Y$.
Define $M:=\max\{m_1,m_2\}$ and 
$m:=\min\{m_1,m_2\}$. 
The definition of $(m_1,m_2)$ implies $Ah^d\le \frac1{M}$,
and hence
\begin{align}
h\le \left(\frac1{AM}\right)^{1/d}.
\label{e174}
\end{align}
Recall $B$ defined in \eqref{e_B}. We then have
\begin{align}
1+\delta
&=m_1m_2P_{XY}(\mathcal{A}\times \mathcal{B})
\\
&\le m_1m_2Bh^{2d}
\\
&\le \frac{Bm_1m_2}{A^2m^2}
\\
&=\frac{BM}{A^2m}.
\label{e_178}
\end{align}

Next, observe that 
$1= m_1 P_X(\mathcal{A}\times[-1,1]^d) \le m_1 Bh^d$ which yields $h^d\ge \frac1{m_1B}$. 
Similarly we also have $h^d\ge \frac1{m_2B}$, therefore
\begin{align}
h^d\ge \frac1{mB}
\label{e178}
\end{align}
Together with \eqref{e174}, we see that $h^d=\Theta(1/m)=\Theta(1/M)$.

Next, the bias of the density estimator is
\begin{align}
\mathbb{E}[\hat{p}^{\rm B}]-p_{XY}(x_0,y_0)
=
\frac{P_{XY}(\mathcal{A}\times \mathcal{B})}{h^{2d}}-p_{XY}(x_0,y_0)
\end{align}
which is just the bias of the rectangular kernel estimator (with bandwidth $h$ in each of the two subspaces).
The rectangle kernel is order $1$ \cite[Definition~1.3]{tsybakov2008introduction} and compactly supported while, by assumption $\beta\in(0,2]$, and therefore the bias is bounded by (\cite[Proposition 1.2]{tsybakov2008introduction})
\begin{align}
|\mathbb{E}[\hat{p}^{\rm B}]-p_{XY}(x_0,y_0)|
\le Ch^{\beta}
\label{e181}
\end{align}
where $C$ is a constant depending only on $\beta,d$ and $L$.

\subsection{One-Way Case}\label{sec_d1way}
By Corollary~\ref{cor7} and \eqref{e178}, we can bound the variance of the density estimator as
\begin{align}
\var(\hat{p}^{\rm B})&=\frac1{m_1^2m_2^2h^{4d}}\var(\hat{\delta}^{\rm B})
\\
&\le \frac1{m_1^2m_2^2h^{4d}}
\cdot \frac{25(1+\delta)m_1m_2}{n}
\\
&\le \frac{25(1+\delta)B^4m^3}{nM}
\label{e184}
\end{align}
where \eqref{e184} used \eqref{e178}.
Also by Corollary~\ref{cor7}, the communication constraint is satisfied if the following holds
\begin{align}
\frac{2.2(1+\delta)n}{m_1}\log_2\frac{m_1}{10} +1\le k.
\end{align}
Now we can choose $h$ so that $m_1=(\frac{k}{\log_2 k})^{\frac{d}{d+2\beta}}$ as defined by \eqref{em1}, and set
\begin{align}
n&=\left\lfloor\left(\frac{2.2(1+\delta_{\rm max})}{m_1}\log_2\frac{m_1}{10}\right)^{-1} (k-1)\right\rfloor
\end{align}
where $\delta_{\rm max}$, defined as the right side of \eqref{e_178} and hence depends only on $(A,\beta,L)$, is an upper estimate of the true parameter $\delta$.
Then the communication constraint is satisfied.
Moreover by the bias \eqref{e181} and the variance \eqref{e184} bounds, 
the risk is bounded by 
\begin{align}
&\quad|\mathbb{E}[\hat{p}^{\rm B}]-p_{XY}(x_0,y_0)|^2
+\var(\hat{p}^{\rm B})
\nonumber\\
&\le C^2h^{2\beta}+\frac{25(1+\delta)B^4m^3}{nM}
\\
&=(Am)^{-2\beta/d}+\frac{25(1+\delta)B^4m^3}{nM}
\label{e189}
\\
&\le D (\frac{k}{\log k})^{-\frac{2\beta}{d+2\beta}}
\end{align}
where $D$ is a constant depending only on $\beta$, $L$,  and $A$, 
and we used the fact that $\delta$ is bounded above by 
\eqref{e178} and the bound on $h$ shown in \eqref{e174}.
This proves the upper bound in Theorem~\ref{thm3} for $\beta\in(0,2]$.

\subsection{Interactive Case}\label{sec_d_interactive}
Choose $h$ such that $m$ as defined by $m:=\min\{m_1,m_2\}$ and \eqref{em1}-\eqref{em2} satisfies 
\begin{align}
m:=k^{\frac{d}{d+2\beta}},
\end{align}
and set
\begin{align}
n:=\left\lfloor\frac{mk\ln2}{13(1+\delta_{\rm max})}\right\rfloor
\end{align}
where as before $\delta_{\rm max}$ is an upper bound on $\delta$ and only depends on $(A,\beta,L)$.
By Corollary~\ref{cor8}, for $k$ large enough we have $m>10$, and the communication cost is bounded by $k$. 
Moreover from \eqref{e189}, the risk is bounded by $Dk^{-\frac{2\beta}{d+2\beta}}$ for some $D$ depending only on $\beta$, $L$, and $A$.
This proves the upper bound in Theorem~\ref{thm4} when $\beta\in(0,2]$.

\subsection{Extension to $\beta>2$}\label{sec_betage}
For $\beta>2$, the rectangular kernel is no longer minimax optimal. 
However, observe the following:
\begin{prop}\label{prop10}
For any positive integers $d$ and $l$, there exists an order $l$-kernel in $\mathbb{R}^d$ which is a linear combination of $(\lfloor l/2\rfloor+1)^d$ indicator functions.
\end{prop}
\begin{proof}
In the following we prove for $d=1$; the case of general $d$ will then follow by taking the tensor product of kernel functions in $\mathbb{R}$.
Note that an $l$-th kernel must satisfy the following equations:
\begin{align}
\int K(u)du&=1;\label{e142}
\\
\int u^jK(u)du&=0,\quad j=1,\dots,l.
\label{e143}
\end{align}
Let use consider $K$ of the following form:
\begin{align}
K(u)=\sum_{k=1}^{k_0}c_k1_{[-k,k]}
\end{align}
where $k_0:=\lfloor l/2\rfloor+1$.
Since such $K(u)$ is an even function, \eqref{e142}-\eqref{e142} yield $k_0$ nontrivial equations for $c_1,\dots, c_{k_0}$ (i.e., only when $j$ is even):
\begin{align}
2\sum_{k=1}^{k_0}kc_k=1;
\\
\sum_{k=1}^{k_0}\frac{2k^{j+1}}{j+1}c_k&=0,\quad j\in\{1,\dots,l\}\cap 2\mathbb{Z}.
\end{align}
From the formula for the determinant of the Vandermonde matrix, we see that these equations have a unique solution for $c_1,\dots,c_{k_0}$.
\end{proof}
Now for general $\beta>0$,
we can take an order $l=\lfloor \beta\rfloor$ kernel as in Proposition~\ref{prop10}.
We can estimate $\frac1{h^{2d}}\int p_{XY}(x,y)K(\frac{(x,y)-(x_0,y_0)}{h})$ by applying the Bernoulli distribution estimator repeatedly for $(\lfloor l/2\rfloor+1)^{2d}$ times.
Therefore by the similar arguments as the preceding sections we see that the upper bounds in Theorem~\ref{thm3} and Theorem~\ref{thm4} hold for $\beta>2$ as well.

\section{One-way Density Estimation Lower Bound}\label{sec_1pl}
\subsection{Upper Bounding $s^*(X;Y)$}
The pointwise estimation lower bound is obtained by lower bounding the risk for estimating $P_{\uX\uY}$ (with $\uX$ and $\uY$ being indicators of neighborhoods of $x_0$ and $y_0$), and applying Le Cam's inequality to the latter.
Therefore we are led to considering the strong data processing constant for the biased Bernoulli distribution.
\begin{thm}\label{thm_1point}
Let $P_{XY}^{(\delta)}$ be as defined in \eqref{e_pdelta}.
where $\delta\in(-1,1)$ and $m>1$.
Then $s^*(X;Y)\le \frac{\delta^2}{m\ln m-m+1}$.
\end{thm}
\begin{proof}
For this proof we can assume without loss of generality that the logarithms are natural.
For any $Q_X$, let $Q_Y$ be the output through the channel $P_{Y|X}$.
Then
\begin{align}
s^*(X;Y)
&\le
\frac{D(Q_Y\|P_Y)}{D(Q_X\|P_X)}
\label{e_wellknown}
\\
&\le \frac{D_{\chi^2}(Q_Y\|P_Y)}{D_{\chi^2}(Q_X\|P_X)}
\cdot\frac{D_{\chi^2}(Q_X\|P_X)}{D(Q_X\|P_X)}
\label{e188}
\\
&\le \frac{\delta^2}{m\ln m-m+1}
\label{e_chi}
\end{align}
where we define the $\chi^2$ divergence in \eqref{e_chi2}.
The justification of the steps are as follows:
\eqref{e_wellknown} is well-known.
\eqref{e188} follows since the $\chi^2$ divergence dominates the KL divergence (see e.g.\ \cite{7552457}).
To see \eqref{e_chi}, note that $\frac{D_{\chi^2}(Q_Y\|P_Y)}{D_{\chi^2}(Q_X\|P_X)}$ is upper bounded by $\rho_{\sf m}^2(X,Y)$, the square of the maximal correlation coefficient (see e.g.\ \cite{ahlswede1976spreading,anantharam2013maximal}).
As the operator norm of a linear operator,
$\rho_{\sf m}(X,Y)$ can be shown to equal the second largest singular value of 
\begin{align}
{\bf M}&:=\left(\frac1{\sqrt{P_X(x)}}P_{XY}(x,y)\frac1{\sqrt{P_Y(y)}}\right)_{x,y}
\label{e190}
\\
&=\begin{bmatrix}
\frac{1+\delta}{m} & *
\\
* &1-\frac1{m}+\frac{\delta}{m(m-1)}
\end{bmatrix};
\label{e191}
\end{align}
see e.g. \cite{anantharam2013maximal}.
Since ${\bf M}$ is a symmetric matrix, its singular values are its eigenvalues.
The largest eigenvalue of ${\bf M}$ is 1, corresponding to the eigenvector $(\sqrt{P_X(0)},\sqrt{P_X(1)})$ (which is evident from \eqref{e190}), 
whereas the trace 
\begin{align}
\tr({\bf M})=\frac{1+\delta}{m}
+1-\frac1{m}+\frac{\delta}{m(m-1)}=1+\frac{\delta}{m-1}
\end{align}
which is evident from \eqref{e191}.
Therefore $\rho_{\rm m}(X;Y)= \frac{\delta}{m-1}$.
Moreover, since $\chi^2$ and $KL$ divergences are both $f$-divergences, their ratio can be bounded by the ratio of their corresponding $f$-functions (see e.g.\ \cite{7552457}):
\begin{align}
\frac{D_{\chi^2}(Q_X\|P_X)}{D(Q_X\|P_X)}
&\le \sup_{0< t\le m}\frac{(t-1)^2}{t\ln t-t+1}
\label{e205}
\\
&= \frac{(m-1)^2}{m\ln m-m+1};
\label{e199}
\end{align}
The constraint $t\le m$ in \eqref{e205} is because $\min_xP_X(x)=\frac1{m}$ and $\max_x\frac{Q_X(x)}{P_X(x)}\le m$.
To show \eqref{e199}, it suffices to show that $$\inf_{u\in(-1,m-1]}\frac{(1+u)\ln(1+u)-u}{u^2}$$ is achieved at $u=m-1$.
For this, it suffices to show that the derivative of the objective function, $\frac{-(2+u)\ln(1+u)+2u}{u^3}$ is negative on $(-1,m-1]$.
Indeed, define $\phi(u):=(2+u)\ln (1+u)-2u$.
We can check that $\phi(0)=0$, $\phi'(0)=0$, and $\phi''(u)=\frac{u}{(1+u)^2}$, which imply that $\phi(u)>0$ for $u>0$ and $\phi(u)<0$ for $u<0$.
Therefore \eqref{e199}, and hence \eqref{e_chi}, holds.
\end{proof}

\subsection{Lower Bounding One-Way NP Estimation Risk}\label{1waylower}
Given $k,d,\beta,L,A$, consider a distribution $P_{\underline{X}\underline{Y}}$ on $\{0,1\}^2$ with matrix
\begin{align}
\left(
\begin{array}{cc}
\frac1{m^2}(1+\delta)  & \frac1{m}(1-\frac1{m})-\frac{\delta}{m^2}      \\
\frac1{m}(1-\frac1{m})-\frac{\delta}{m^2}
&     (1-\frac1{m})^2+\frac{\delta}{m^2}
\end{array}
\right),
\label{e196}
\end{align}
where
$m:=\left(\frac{ak}{\ln k}\right)^{\frac{d}{2\beta+d}}$ and $\delta:=m^{-\frac{\beta}{d}}$,
with $a:=\frac{16\beta+8d}{d}\ln2$ being a constant. 
We then need to ``simulate'' smooth distributions from $P_{\underline{X}\underline{Y}}$.
Let $f\colon \mathbb{R}^d\to [0,\infty)$ be a function satisfying the following properties:
\begin{itemize}
\item $f$ has a compact support;
\item $\int_{\mathbb{R}^d} f=1$;
\item $f(0)>0$;
\item $f(x)\in [0,1]$, for all $x\in\mathbb{R}^d$;
\item $\|f\|_{\mathbb{R}^d,\beta}<\frac{L}{4}$;
\item Define $g(x,y)=f(x)f(y)$ as a function on $\mathbb{R}^{2d}$. Then $\|g\|_{\mathbb{R}^{2d},\beta}<\frac{L}{4}$.
\end{itemize}
Clearly, such a function exists for any given $\beta,L,d$.
For sufficiently large $m$ such that $m^{-1/d}\supp(f)+x_0\in[0,1]^d$ and $m^{-1/d}\supp(f)+y_0\in[0,1]^d$ (recall that $(x_0,y_0)$ is the given point in the density estimation problem),
define 
\begin{align}
p_{X|\underline{X}=0}(x):=
\frac1{P_{\underline{X}}(0)}f(m^{\frac1{d}}(x-x_0)),\quad \forall x\in\mathbb{R}^d;
\\
p_{Y|\underline{Y}=0}(y):=
\frac1{P_{\underline{Y}}(0)}f(m^{\frac1{d}}(y-y_0)),\quad \forall y\in\mathbb{R}^d.
\end{align}
Since $P_{\underline{X}}(0)=P_{\underline{Y}}(0)=\frac1{m}$, clearly the above define valid probability densities supported on $[0,1]^d$.
Define 
\begin{align}
p_{X|\underline{X}=1}(x):=\frac{1\{x\in[0,1]^d\}}{P_{\underline{X}}(1)}[1-f(m^{1/d}(x-x_0))];
\\
p_{Y|\underline{Y}=1}(y):=\frac{1\{y\in[0,1]^d\}}{P_{\underline{Y}}(1)}[1-f(m^{1/d}(y-y_0))],
\end{align}
which are also probability densities supported on $[0,1]^d$.
Define $P_{XY|\underline{XY}}=P_{X|\underline{X}}P_{Y|\underline{Y}}$, where $P_{X|\underline{X}}$ and $P_{Y|\underline{Y}}$ are conditional distributions defined by the densities above.
Under the joint distribution $P_{XY\underline{XY}}$, we have 
\begin{align}
p_X(x)=p_Y(y)=1,\quad\forall x,y\in[0,1]^d.
\end{align}
Define 
\begin{align}
\bar{P}_{XY\underline{XY}}=P_{X|\underline{X}}P_{Y|\underline{Y}}P_{\underline{X}}P_{\underline{Y}}.
\end{align}
We now check that the density of $P_{XY}$ satisfies $\|p_{XY}\|_{(0,1)^{2d},\beta}\le L$ for $m$ sufficiently large.
Indeed, for $x,y\in[0,1]^d$,
\begin{align}
&\quad p_{XY}(x,y)
\nonumber\\
&=\sum_{i,j=0,1}p_{XY|\underline{XY}=(i,j)}(x,y)P_{\underline{XY}}(i,j)
\\
&=\sum_{i,j=0,1}p_{XY|\underline{XY}=(i,j)}(x,y)\bar{P}_{\underline{XY}}(i,j)
\nonumber\\
&\quad+\tfrac{\delta}{m^2}(p_{X|\underline{X}=0}(x)-p_{X|\underline{X}=1}(x))(p_{Y|\underline{Y}=0}(y)-p_{Y|\underline{Y}=1}(y))
\\
&=1+\tfrac{\delta}{m^2}(p_{X|\underline{X}=0}(x)-p_{X|\underline{X}=1}(x))
\nonumber\\
&\quad\cdot(p_{Y|\underline{Y}=0}(y)-p_{Y|\underline{Y}=1}(y))
\\
&=1+\delta\left[-\tfrac1{m-1}+\tfrac1{1-\tfrac1{m}}f(m^{1/d}(x-x_0))\right]
\nonumber\\
&\quad\cdot\left[-\tfrac1{m-1}+\tfrac1{1-\tfrac1{m}}f(m^{1/d}(y-y_0))\right]
\label{e206}
\\
&={\rm const.}-\frac{\delta m}{(m-1)^2}f(m^{1/d}(x-x_0))
\nonumber\\
&\quad-\frac{\delta m}{(m-1)^2}f(m^{1/d}(y-y_0))
\nonumber\\
&\quad+\frac{\delta}{(1-1/m)^2}f(m^{1/d}(x-x_0))f(m^{1/d}(y-y_0)).
\end{align}
By the assumptions on $f$, we see that 
\begin{align}
\|m^{-\beta/d}f(m^{1/d}(\cdot-x_0))\|_{(0,1)^d,\beta}&\le \frac{L}{4};
\\
\|m^{-\beta/d}f(m^{1/d}(\cdot-y_0))\|_{(0,1)^d,\beta}&\le \frac{L}{4};
\\
\|m^{-\beta/d}f(m^{1/d}(\cdot-x_0))f(m^{1/d}(*-y_0))\|_{(0,1)^{2d},\beta}&\le \frac{L}{4}.
\end{align}
Therefore with the choice $\delta=m^{-\beta/d}$, we have $\|p_{XY}\|_{(0,1)^{2d},\beta}\le L$ for $m\ge 10$.

Now we can apply a Le Cam style argument \cite{assouad1996fano} for the estimation lower bound.
Suppose that there exists an algorithm that estimates the density at $(x_0,y_0)$ as $\hat{p}$.
Alice and Bob can convert this to an algorithm for testing the binary distributions $P_{\underline{XY}}$ against $\bar{P}_{\underline{XY}}$.
Indeed, suppose that ${\bf (\underline{X},\underline{Y})}$ is an infinite sequence of i.i.d.\ random variable pairs according to either $P_{\underline{XY}}$ or $\bar{P}_{\underline{XY}}$.
Using the local randomness (which is implied by the common randomness),
Alice and Bob can simulate the sequence of i.i.d.\ random variables ${\bf (X,Y)}$ according to either $P_{XY}$ or $\bar{P}_{XY}$, by applying the random transformations $P_{X|\underline{X}}$ and $P_{Y|\underline{Y}}$ coordinate-wise.
Then Alice and Bob can apply the density estimation algorithm to obtain $\hat{p}$.
Note that $\bar{p}_{XY}(x_0,y_0)=1$ while 
\begin{align}
p_{XY}(x_0,y_0)=1+\delta\left[
\frac{m}{m-1}f(0)-\frac1{m-1}
\right]^2,
\label{e216}
\end{align}
the latter following from \eqref{e206}.
Now suppose that Bob declares $P_{\underline{XY}}$ if 
\begin{align}
|\hat{p}-p_{XY}(x_0,y_0)|\le |\hat{p}-1|,
\label{e_test}
\end{align}
and $\bar{P}_{\underline{XY}}$ 
otherwise.
By Chebyshev's inequality, the error probability (under either hypothesis) is upper bounded by 
\begin{align}
4\delta^{-2}
\left[
\tfrac{m}{m-1}f(0)-\tfrac1{m-1}
\right]^{-4}
\hspace{-10pt}\sup_{p_{XY}\in\mathcal{H}(\beta,L,A)}\mathbb{E}[|\hat{p}-p_{XY}(x_0,y_0)|^2].
\label{e211}
\end{align}
On the other hand, from \eqref{e19} and Theorem~\ref{thm_1point} we have
\begin{align}
\frac{D(P_{{\bf \underline{Y}}\Pi}\|\bar{P}_{{\bf \underline{Y}}\Pi})}{H(\Pi)}&\le s^*(\underline{X};\underline{Y})
\label{e213}
\\
&\le \frac{\delta^2}{m\ln m-m+1}
\\
&\le \frac{2\delta^2}{m\cdot\frac{d}{4\beta+2d}\ln k}
\\
&\le \frac{8\beta+4d}{dak}
\end{align}
when $m$ is sufficiently large.
However, it is known (from Kraft's inequality, see e.g.\ \cite{elements}) that the expected length of a prefix code upper bounds the entropy. Thus
\begin{align}
k\ge \mathbb{E}[|\Pi|]\ge \frac1{\log 2}H(\Pi)
\end{align}
and therefore
\begin{align}
D(P_{{\bf \underline{Y}}\Pi}\|\bar{P}_{{\bf \underline{Y}}\Pi})\le 
\frac{8\beta+4d}{da}\log 2.
\end{align}
Then by Pinsker's inequality (e.g.\ \cite{tsybakov2008introduction}), 
\begin{align}
1-\int dP_{{\bf \underline{Y}}\Pi}\wedge d\bar{P}_{{\bf \underline{Y}}\Pi} &\le \sqrt{\frac1{2\log e}D(P_{{\bf \underline{Y}}\Pi}\|\bar{P}_{{\bf \underline{Y}}\Pi})}
\\
&\le \sqrt{\frac{4\beta+2d}{da}\ln 2}
\\
&=\frac1{2}
\end{align}
where the last line follows from our choice $a=\frac{16\beta+8d}{d}\ln2$.
However, $\int dP_{{\bf \underline{Y}}\Pi}\wedge d\bar{P}_{{\bf \underline{Y}}\Pi}$ lower bounds twice of \eqref{e211}.
Therefore we have
\begin{align}
&\quad\sup_{p_{XY}\in\mathcal{H}(\beta,L,A)}\mathbb{E}[|\hat{p}-p_{XY}(x_0,y_0)|^2]
\nonumber\\
&\ge
\frac{\delta^2}{8}
\left[
\frac{m}{m-1}f(0)-\frac1{m-1}
\right]^4
\cdot\frac1{2}
\\
&=\frac1{16}m^{-2\beta/d}\left[
\frac{m}{m-1}f(0)-\frac1{m-1}
\right]^4
\end{align}
which is lower bounded by $\frac1{17}m^{-2\beta/d}f^4(0)=\frac{f^4(0)}{17}\left(\frac{ak}{\ln k}\right)^{-\frac{2\beta}{2\beta+d}}$ for large enough $k$.
Since $a$ and $f(0)$ depend only on $d,\beta,L$, this 
establishes the lower bound in Theorem~\ref{thm3}.

\section{Interactive Density Estimation Lower Bound}\label{app9}
In this section we prove the lower bound in Theorem~\ref{thm4}.
\subsection{Upper Bounding $s^*_{\infty}(X;Y)$}
The heart of the proof is the following technical result:
\begin{thm}\label{thm7}
There exists $c>0$ small enough such that the following holds: For any $P_{XY}$ which is a distribution on $\{0,1\}^2$ corresponding to the following matrix:
\begin{align}
\left(
\begin{array}{cc}
p^2(1+\delta)  & p\bar{p}-p^2\delta      \\
p\bar{p}-p^2\delta
& \bar{p}^2+p^2\delta
\end{array}
\right)
\label{e138}
\end{align}
where $p,|\delta|\in [0,c)$ and we used the notation $\bar{p}:=1-p$, we have
\begin{align}
s_{\infty}^*(X;Y)\le c^{-1}p\delta^2.
\end{align}
\end{thm}
The proof can be found in Appendix~\ref{app_thm7}.

\subsection{Lower Bounding Interactive NP Estimation Risk}
The proof is similar to the one-way case (Section~\ref{1waylower}).
Consider the distribution $P_{\underline{XY}}$
on $\{0,1\}^2$ as in \eqref{e196}.
Let $m:=(ak)^{\frac{d}{2\beta+d}}$ and $\delta:=m^{-\frac{\beta}{d}}$, where $a=\frac{2\ln 2}{c}$ with $c$ being the absolute constant in Theorem~\ref{thm7}.
Pick the function $f$, and define $P_{\underline{XY}XY}$ and $\bar{P}_{\underline{XY}XY}$ as before.
Note that, as before, $\bar{p}_{XY}$ is uniform on $[0,1]^{2d}$, while
$\|p_{XY}\|_{(0,1)^{2d},\beta}\le L$ for $m\ge 10$.
$p_{XY}(x_0,y_0)$ has the same formula \eqref{e216},
and Alice and Bob can convert a (now interactive) density estimation algorithm to an algorithm for testing $P_{\underline{XY}}$ against $\bar{P}_{\underline{XY}}$.
With the same testing rule \eqref{e_test}, the error probability under either hypothesis is again upper bounded by \eqref{e211}.

Changes arise in \eqref{e213}, where we shall apply Theorem~\ref{thm7} instead.
Note that for the absolute constant $c$ in Theorem~\ref{thm7}, the condition $\frac1{m},|\delta|<c$ is satisfied for sufficiently large $k$ (hence sufficiently large $m$).
\begin{align}
\frac{D(P_{{\bf \underline{Y}}\Pi}\|\bar{P}_{{\bf \underline{Y}}\Pi})}{H(\Pi)}
&\le s^*_{\infty}(\underline{X};\underline{Y})
\\
&\le \frac{c^{-1}\delta^2}{m}
\\
&\le (cak)^{-1}.
\end{align}
Again using Kraft's inequality to Bound $H(\Pi)$, we obtain
\begin{align}
D(P_{{\bf \underline{Y}}\Pi}\|\bar{P}_{{\bf \underline{Y}}\Pi})\le \frac{\log 2}{ca}.
\end{align}
Then Pinsker's inequality yields 
\begin{align}
1-\int dP_{\underline{\bf Y}\Pi}\wedge d\bar{P}_{\underline{\bf Y}\Pi}
\le \sqrt{\frac{\ln 2}{2ca}}=\frac1{2}
\end{align}
since we selected $a=\frac{2\ln 2}{c}$.
Again $\int dP_{\underline{\bf Y}\Pi}\wedge d\bar{P}_{\underline{\bf Y}\Pi}$ lower bounds twice of \eqref{e211},
therefore
\begin{align}
&\quad\sup_{p_{XY}\in\mathcal{H}(\beta,L,A)}\mathbb{E}[|\hat{p}-p_{XY}(x_0,y_0)|^2]
\nonumber\\
&\ge
\frac{\delta^2}{8}
\left[
\frac{m}{m-1}f(0)-\frac1{m-1}
\right]^4
\cdot\frac1{2}
\\
&=\frac1{16}m^{-2\beta/d}\left[
\frac{m}{m-1}f(0)-\frac1{m-1}
\right]^4
\\
&\ge \frac{f^4(0)}{17}(ak)^{-\frac{2\beta}{2\beta+d}}
\end{align}
where the last line holds for sufficiently large $k$. Since $a$ is a universal constant and $f$ depends on $d,\beta,L$ only, this completes the proof of the interactive lower bound.

\section{Acknowledgement}
The author would like to thank Professor Venkat Anantharam for bringing the reference 
\cite{orlitsky1990worst} to my attention and some interesting discussions.
This research was supported by the starting grant from the Department of Statistics, University of Illinois, Urbana-Champaign.

\appendices
\section{Proof of Theorem~\ref{thm_s}}
\label{app_thm_s}
First, assume that $i\in\{1,2,\dots, r\}\setminus 2\mathbb{Z}$.
By the definitions of $P_{U_i|XU^{i-1}}$ and $P_{U_i|YU^{i-1}}$, we can verify that the following holds (for $\delta=0$):
\begin{align}
P_{X|U^i}^{(0)}(0|{\bf 0})=\frac{P_X(0)}{P_X(1)\prod_{1\le j\le i}^{\rm odd}\alpha_j^{-1}+P_X(0)}.
\label{e50}
\end{align}
Indeed, \eqref{e50} follows by applying induction on the following
\begin{align}
P^{(0)}_{X|U^i}(0|{\bf 0})
&=\frac{p_0}{p_0+p_1}
\\
&=\frac{P^{(0)}_{X|U^{i-2}}(0|{\bf 0})}{P^{(0)}_{X|U^{i-2}}(0|{\bf 0})+
P^{(0)}_{X|U^{i-2}}(1|{\bf 0})\alpha_i^{-1}}
\end{align}
where $p_0:=P^{(0)}_{X|U^{i-1}}(0|{\bf 0})P^{(0)}_{U_i|XU^{i-1}}(0|0,{\bf 0})$,
$p_1:=
P^{(0)}_{X|U^{i-1}}(1|{\bf 0})P^{(0)}_{U_i|XU^{i-1}}(0|1,{\bf 0})$,
and we used $P^{(0)}_{X|U^{i-1}}(0|{\bf 0})=P^{(0)}_{X|U^{i-2}}(0|{\bf 0})$ which in turn follows from the factorization
\begin{align}
P^{(0)}_{XYU_{i-1}|U^{i-2}}
&=
P^{(0)}_{XY|U^{i-2}}P^{(0)}_{U_{i-1}|YU^{i-2}}
\\
&=P^{(0)}_{X|U^{i-2}}P^{(0)}_{Y|U^{i-2}}P^{(0)}_{U_{i-1}|YU^{i-2}}.
\end{align}
Now from \eqref{e50},
\begin{align}
P^{(0)}_{U_i|U^{i-1}}(0|{\bf 0})&=
P^{(0)}_{X|U^{i-1}}(0|{\bf 0})+\alpha_i^{-1}P^{(0)}_{X|U^{i-1}}(1|{\bf 0})
\\
&=\frac{P_X(1)\prod_{1\le j\le i}^{\rm odd}\alpha_j^{-1}+P_X(0)}{P_X(1)\prod_{1\le j\le i-2}^{\rm odd}\alpha_j^{-1}+P_X(0)}.
\label{e52}
\end{align}
Similarly, by switching the roles of $X$ and $Y$ we have
\begin{align}
P^{(0)}_{U_{i-1}|U^{i-2}}(0|{\bf 0})
&=\frac{P_Y(1)\prod_{2\le j\le i-1}^{\rm even}\alpha_j^{-1}+P_Y(0)}{P_Y(1)\prod_{2\le j\le i-3}^{\rm even}\alpha_j^{-1}+P_Y(0)}.
\label{e54}
\end{align}
Therefore,
\begin{align}
&\quad P_{U^i}^{(0)}({\bf 0})
\nonumber\\
&=\prod_{1\le j\le i}P_{U_i|{\bf U}^{i-1}}^{(0)}(0|{\bf 0})
\\
&=\left[P_X(1)\prod_{1\le j\le i}^{\rm odd}\alpha_j^{-1}+P_X(0)\right]
\left[P_Y(1)\prod_{2\le j\le i}^{\rm even}\alpha_j^{-1}+P_Y(0)\right]
\label{e_supplement62}
\end{align}
for any $i=1,\dots,r$.
We also see from
\eqref{e50} and \eqref{e_supplement62} that for any $i$ odd,
\begin{align}
P^{(0)}_{XU^{i-1}(0,{\bf 0})}
&=P_X(0)\left[P_Y(1)\prod_{2\le j\le i-1}^{\rm even}\alpha_j^{-1}+P_Y(0)
\right]
\\
&=\frac1{m_1}\left[(1-\frac1{m_2})\prod_{2\le j\le i-1}^{\rm even}\alpha_j^{-1}+\frac1{m_2}
\right]
\\
&\le \frac{1.1}{m_1} \prod_{2\le j\le i-1}^{\rm even}\alpha_j^{-1}
\end{align}
where the last step follows since $\prod_{2\le j\le i-1}^{\rm even}\alpha_j^{-1}\ge \frac{10}{m_2}$.
Therefore, the claim \eqref{e43} follows, and \eqref{e44} is similar.

Next, 
consider any $i\in\{1,\dots,r\}$. 
Define
\begin{align}
\delta_i:=\frac{P_{XY|U^{i-1}={\bf 0}}(0,0)}{P_{X|U^{i-1}={\bf 0}}(0)P_{Y|U^{i-1}={\bf 0}}(0)}-1.
\end{align}
Observe that the construction of $P_{U^r|XY}$ fulfills the Markov chain conditions \eqref{e_markov0}-\eqref{e_markov}, implying  that
\begin{align}
\frac{P^{(\delta)}_{XY}(0,0)P^{(\delta)}_{XY}(1,1)}{P^{(\delta)}_{XY}(0,1)P^{(\delta)}_{XY}(1,0)}
=\frac{P^{(\delta)}_{XY|U^{i-1}}(0,0|{\bf 0})P^{(\delta)}_{XY|U^{i-1}}(1,1|{\bf 0})}{P^{(\delta)}_{XY|U^{i-1}}(0,1|{\bf 0})P^{(\delta)}_{XY|U^{i-1}}(1,0|{\bf 0})}.
\end{align}
We therefore have\footnote{We use the notation $\bar{x}:=1-x$ for $x\in[0,1]$.} 
\begin{align}
\frac{(1+\delta)(1+\frac{\delta}{(m_1-1)(m_2-1)})}{(1-\frac{\delta}{m_1-1})(1-\frac{\delta}{m_2-1})}
=
\frac{(1+\delta_i)(1+\delta_i\frac{b^{(\delta)}c^{(\delta)}}{\bar{b}^{(\delta)}\bar{c}^{(\delta)}})}
{(1-\delta_i\frac{b^{(\delta)}}{\bar{b}^{(\delta)}})(1-\delta_i\frac{c^{(\delta)}}{\bar{c}^{(\delta)}})}
\label{e70}
\end{align}
where we defined
\begin{align}
b^{(\delta)}&:=P_{X|U^{i-1}}^{(\delta)}(0|{\bf 0});
\\
c^{(\delta)}&:=P_{Y|U^{i-1}}^{(\delta)}(0|{\bf 0}).
\end{align}
By continuity, we have
\begin{align}
b^{(\delta)}&=b^{(0)}+o(1);
\label{e71}
\\
c^{(\delta)}&=c^{(0)}+o(1),
\label{e72}
\end{align}
as $\delta\to 0$. 
It is also easy to see from \eqref{e70} that $\delta_i=O(\delta)$ (for this proof, only $\delta$ is the variable, and all other constants, such as $m$ and $(\alpha_i)$, can be hidden in the Landau notations).
Therefore \eqref{e70}\eqref{e71}\eqref{e72} yield
\begin{align}
&\quad1+\left(1+\frac{1}{m_1-1}\right)
\left(1+\frac{1}{m_2-1}\right)
\delta+o(\delta)
\nonumber\\
&=
1+\left(1+\frac{b^{(\delta)}}{\bar{b}^{(\delta)}}
\right)
\left(1+\frac{c^{(\delta)}}{\bar{c}^{(\delta)}}
\right)
\delta_i+o(\delta)
\\
&=
1+\left(1+\frac{b^{(0)}}{\bar{b}^{(0)}}
\right)
\left(1+\frac{c^{(0)}}{\bar{c}^{(0)}}
\right)
\delta_i+o(\delta).
\label{e_62}
\end{align}
Using the fact that $X$ and $Y$ are independent under $P^{(0)}$, noting \eqref{e50} and the assumption $\prod_{1\le j\le r}^{\rm odd}\alpha_j^{-1}\ge \frac{10}{m_1}$, we have 
\begin{align}
b^{(0)}=P_{X|U^{i'}}^{(0)}(0|{\bf 0})
\le \frac{\frac1{m_1}}{\frac{10}{m_1}(1-\frac1{m_1})+\frac1{m_1}}\le \frac1{10}, 
\end{align}
where $i'$ is the largest odd integer not exceeding $i$.
Similarly we also have $c^{(0)}\le \frac1{10}$.
Consequently, \eqref{e_62} yields
\begin{align}
\delta_i^2
&\ge \delta^2\left(\frac{(1+\frac1{m_1-1})(1+\frac1{m_2-1})}{(1+\frac1{9})^2}\right)^2+o(\delta^2)
\\
&\ge 0.9^4\delta^2+o(\delta^2).
\end{align}
Moreover,
let us define
\begin{align}
a^{(\delta)}:=P^{(\delta)}_{U_i|U^{i-1}}(0|{\bf 0}).
\end{align}

In the following paragraph we consider arbitrary $i\in\{1,2,\dots,r\}\setminus 2\mathbb{Z}$, and we shall omit the superscripts $(\delta)$ for $a^{(\delta)}$, $b^{(\delta)}$, $c^{(\delta)}$, unless otherwise noted. 
Then
\begin{align}
&\quad I(U_i;Y|U^{i-1}={\bf 0})
\nonumber\\
&=aD(P_{Y|U^i={\bf 0}}\|P_{Y|U^{i-1}={\bf0}})
\nonumber\\
&\quad+
\bar{a}D(P_{Y|U_i=1,U^{i-1}={\bf0}}\|P_{Y|U^{i-1}={\bf 0}})\label{e93}
\end{align}
We can verify that $P_{X|U^i={\bf0}}(0)=\frac{b}{b+\alpha_i^{-1}\bar{b}}$.
Therefore, 
\begin{align}
P_{Y|U^i={\bf 0}}(0)&=
\frac{b}{b+\alpha_i^{-1}\bar{b}}
\cdot c(1+\delta_{i-1})
\nonumber\\
&\quad
+\frac{\alpha_i^{-1}\bar{b}}{b+\alpha_i^{-1}\bar{b}}
\cdot
\frac{\bar{b}c-\delta_{i-1}bc}{\bar{b}}
\\
&=c+\frac{bc(1-\alpha_i^{-1})}{b+\alpha_i^{-1}\bar{b}}\delta_{i-1}.
\label{e95}
\end{align}
Therefore as $\delta\to0$,
\begin{align}
D(P_{Y|U^i={\bf 0}}\|P_{Y|U^{i-1}={\bf 0}})
&=d\left(P_{Y|U^i={\bf 0}}(0)\|c\right)
\\
&=\frac1{2} c\left(
\frac{b(\alpha_i-1)}{\alpha_ib+\bar{b}}\delta_{i-1}
\right)^2+o(\delta^2)
\\
&\ge  \frac1{2}c\left(b(\alpha_i-1)0.9^2\delta\right)^2
+o(\delta^2)
\\
&\ge\frac{0.9^4}{2}cb^2(\alpha_i-1)\delta^2+o(\delta^2)
\end{align}
where $d(p\|q):=p\log\frac{p}{q}+(1-p)\log\frac{1-p}{1-q}$ denotes the binary divergence function, and recall that we assumed the natural base of logarithms.
On the other hand, $P_{Y|U_i=1,U^{i-1}={\bf0}}(0)=P_{Y|X=1,U^{i-1}={\bf0}}(0)=c-\frac{\delta_{i-1}bc}{1-b}$.
Therefore
\begin{align}
D(P_{Y|U_i=1,U^{i-1}={\bf0}}\|P_{Y|U^{i-1}={\bf0}})
&=
d\left(c-\frac{\delta_{i-1}bc}{1-b}\|c\right)
\\
&= \frac1{2} c\left(\frac{\delta_{i-1}b}{1-b}\right)^2+o(\delta^2)
\\
&\ge \frac{0.9^4}{2}cb^2\delta^2+o(\delta^2).
\end{align}
Turning back to \eqref{e93}, we obtain
\begin{align}
&\quad I(U_i;Y|U^{i-1}={\bf 0})
\nonumber\\
&=\frac{0.9^4}{2}\left[ab^2c(\alpha_i-1)^2\delta^2+(1-a)cb^2\delta^2\right]
+o(\delta^2)
\\
&\ge \frac{0.9^5}{2} \left(\alpha_i^{-1}(\alpha_i-1)^2+1-\alpha_i^{-1}\right)b^2c\delta^2
+o(\delta^2)
\label{e80}
\\
&\ge \frac{0.9^5}{2}(\alpha_i-1)b^2c\delta^2
+o(\delta^2)
\label{e104}
\end{align}
where \eqref{e80} follows since \eqref{e52} implies
\begin{align}
a^{(\delta)}&=a^{(0)}+o(1)
\\
&= \frac{(1-\frac1{m_1})\prod_{1\le j\le i}^{\rm odd}\alpha_j^{-1}+\frac1{m_1}}{(1-\frac1{m_1})\prod_{1\le j\le i-2}^{\rm odd}\alpha_j^{-1}+\frac1{m_1}}+o(1)
\\
&\ge 
\frac{(1-\frac1{m_1})\prod_{1\le j\le i}^{\rm odd}\alpha_j^{-1}}{(1-\frac1{m_1})\prod_{1\le j\le i-2}^{\rm odd}\alpha_j^{-1}}+o(1)
\\
&= \alpha_i^{-1}+o(1)
\end{align} 
and 
\begin{align}
1-a^{(\delta)}&=1-a^{(0)}+o(1)
\\
&= \frac{(1-\frac1{m_1})(1-\alpha_i^{-1})\prod_{1\le j\le i-2}^{\rm odd}\alpha_j^{-1}}{(1-\frac1{m_1})\prod_{1\le j\le i-2}^{\rm odd}\alpha_j^{-1}+\frac1{m_1}}
+o(1)
\\
&\ge \frac{(1-\frac1{m_1})(1-\alpha_i^{-1})\frac{10}{m_1}}{(1-\frac1{m_1})\frac{10}{m_1}+\frac1{m_1}}+o(1)
\\
&\ge 0.9(1-\alpha_i^{-1})+o(1).
\end{align}
Moreover, by \eqref{e50},
\begin{align}
b^{(\delta)}&=b^{(0)}+o(1)
\\
&=\frac{\frac1{m_1}}{(1-\frac1{m_1})\prod_{1\le j\le i-1}^{\rm odd}\alpha_j^{-1}+\frac1{m_1}}+o(1)
\\
&\ge 
\frac{\frac1{m_1}}{(1-\frac1{m_1}+\frac1{10})\prod_{1\le j\le i-1}^{\rm odd}\alpha_j^{-1}}+o(1)
\\
&\ge \frac{0.9}{m_1}\prod_{1\le j\le i-1}^{\rm odd}\alpha_j+o(1).
\end{align}
Similarly,
\begin{align}
c^{(\delta)}\ge \frac{0.9}{m_2}\prod_{1\le j\le i-1}^{\rm even}\alpha_j+o(1).
\end{align}
Therefore by \eqref{e104} and \eqref{e_supplement62},
\begin{align}
I(U_i;Y|U^{i-1})
&=
I(U_i;Y|U^{i-1}={\bf 0})P_{U^{i-1}}({\bf 0})
\\
&\ge \frac{0.9^5}{2}(\alpha_i-1)b^2c\delta^2\prod_{1\le j\le i-1}\alpha_j^{-1}
+o(\delta^2)
\\
&\ge \frac{0.9^8\delta^2}{2m_1^2m_2}(\alpha_i-1)\prod_{1\le j\le i-1}^{\rm odd}\alpha_j+o(\delta^2)
\\
&=\frac{0.9^8\delta^2}{2m_1^2m_2}\left(
\prod_{1\le j\le i}^{\rm odd}\alpha_j
-
\prod_{1\le j\le i-2}^{\rm odd}\alpha_j\right)+o(\delta^2)
\end{align}
and hence
\begin{align}
\sum_{1\le i\le r}^{\rm odd}I(U_i;Y|U^{i-1})
&\ge\frac{0.9^8\delta^2}{2m_1^2m_2}\prod_{1\le j\le r}^{\rm odd}\alpha_j
+o(\delta^2),
\end{align}
establishing the claim \eqref{e47} of the theorem.
The proof of \eqref{e48} is similar.

\section{Proof of Theorem~\ref{thm7}}\label{app_thm7}
We can choose the natural base of logarithms in this proof.
Choose ${\bf U}=(U_1,U_2,\dots,U_r)$ satisfying the Markov chain conditions \eqref{e_markov0}-\eqref{e_markov}
and so that
\begin{align}
s_{\infty}^*(X;Y)
&\le 2\cdot\frac{I(X;Y)-I(X;Y|\bf U)}{I({\bf U};X,Y)}
\\
&\le 
4\cdot\frac{I(X;Y)-I(X;Y|\bf U)}{I({\bf U};X)+I({\bf U};Y)}
\end{align}
which is possible by the definition of $s_{\infty}^*(X;Y)$.

Given $\alpha,\beta\in[0,1]$, define by $P^{\alpha,\beta}$ the unique distribution\footnote{Alternatively, $P^{\alpha,\beta}$ equals the $I$-projection $\argmin_{Q_{XY}}D(Q_{XY}\|P_{XY})$ under the constraints $Q_X=[\alpha,\bar{\alpha}]$ and $Q_Y=[\beta,\bar{\beta}]$ \cite[Corollary~3.3]{csiszar1975}.} such that 
\begin{align}
P^{\alpha,\beta}(x,y)=P_{XY}(x,y)f(x)g(y)
\label{e_scaling}
\end{align} 
for some functions $f$ and $g$,
and such that the marginals are $P^{\alpha}:=[\alpha,\bar{\alpha}]$ and $P^{\beta}:=[\beta,\bar{\beta}]$.
For the existence of $P^{\alpha,\beta}$, see e.g.\ \cite{hobby1965doubly,sinkhorn1967diagonal}.
Define $I(\alpha,\beta)$ as the mutual information 
of $(X,Y)$ under $P^{\alpha,\beta}$.
Define $\lambda=\lambda(\alpha,\beta)$ as the number such that $P^{\alpha,\beta}$ is  the matrix
\begin{align}
\left(
\begin{array}{cc}
\alpha\beta+\lambda & \alpha\bar{\beta}-\lambda     \\
\bar{\alpha}\beta-\lambda  
&     \bar{\beta}\bar{\beta}+\lambda
\end{array}
\right).
\label{e_244}
\end{align}

Given any $\bf u$, let $\alpha_{\bf u}\in [0,1]$ be such that $P_{X|{\bf U=u}}=[\alpha_{\bf u},\bar{\alpha}_{\bf u}]$. 
Define $\beta_{\bf u}$ similarly but for $P_{Y|{\bf U=u}}$.
With these notations, note that
\begin{align}
\mathbb{E}[\alpha_{\bf U}]=\mathbb{E}[\beta_{\bf U}]=p;
\label{e_ealpha}
\end{align}
and 
\begin{align}
I(X;Y)-I(X;Y|{\bf U})=I(p,p)-\mathbb{E}[I(\alpha_{\bf U},\beta_{\bf U})];
\\
I({\bf U};X)+I({\bf U};Y)=\mathbb{E}[d(\alpha_{\bf U}\|p)+d(\beta_{\bf U}\|p)]
\end{align}
where we recall that $d(\cdot\|\cdot)$ denotes the binary divergence function.
Define 
\begin{align}
\psi(\alpha,\beta):=d(\alpha\|p)+d(\beta\|p).
\end{align}
Then note that $\psi(\alpha,\beta)$ is a smooth nonnegative function on $[0,1]^2$ with vanishing value and first derivatives at $(p,p)$.
Also, define 
\begin{align}
&\quad\phi(\alpha,\beta):=
\nonumber\\
&I(p,p)-I(\alpha,\beta)+I_{\alpha}(p,p)(\alpha-p)
+I_{\beta}(p,p)(\beta-p)
\end{align}
where $I_{\alpha}(p,p):=\left.\frac{\partial}{\partial\alpha}I(\alpha,\beta)\right|_{(p,p)}$.
Then $\phi$ is also a smooth function on $[0,1]^2$ with vanishing value and first derivatives at $(p,p)$.
Moreover, due to \eqref{e_ealpha}, we have
\begin{align}
\mathbb{E}[\phi(\alpha_{\bf U},\beta_{\bf U})]=I(p,p)-\mathbb{E}[I(\alpha_{\bf U},\beta_{\bf U})].
\end{align}
Thus to prove the theorem it suffices to show the existence of sufficiently small $c>0$, such that for any $p,|\delta|\in (0,c)$, there is
\begin{align}
\sup_{\alpha,\beta}\frac{\phi(\alpha,\beta)}{\psi(\alpha,\beta)}
\le c^{-1}p\delta^2
\label{e_suff}
\end{align}
where the sup is over $\alpha,\beta\in (0,1)$.
\begin{itemize}
\item {\bf Case~1:} $0.1p<\alpha,\beta<10p$.\\
Since $\frac{\partial^2}{\partial \alpha^2}D(\alpha\|p)= \left[\frac1{\alpha}+\frac1{1-\alpha}\right]\ge \frac1{\alpha}\ge \frac1{10 p}$ for $\alpha\in[0,10p]$,
we have
\begin{align}
\psi(\alpha,\beta)\ge  \frac1{20p}[(\alpha-p)^2+(\beta-p)^2]
\end{align}
for $(\alpha,\beta)\in[0,10p]^2$.
Now if we can show that 
\begin{align}
\sup_{(\alpha,\beta)\in[0,10p]^2}\|\partial^2\phi(\alpha,\beta)\|
&=
\sup_{(\alpha,\beta)\in[0,10p]^2}\|\partial^2I(\alpha,\beta)\|
\\
&\lesssim \delta^2,
\label{e_i2d}
\end{align}
we will obtain $\sup_{(\alpha,\beta)\in[0,10p]^2}\frac{\phi(\alpha,\beta)}{\psi(\alpha,\beta)}
\lesssim p\delta^2$ which matches \eqref{e_suff}.
Here and below, $x\lesssim y$ means that there is an absolute constant $C>0$ such that 
$x\le Cy$ when $c$ in the theorem statement (and hence $p$ and $|\delta|$) is sufficiently small.

Before explicitly computing $\partial^2\phi(\alpha,\beta)$, 
we give some intuitions why we should expect \eqref{e_i2d} to be true.
For fixed $\alpha,\beta,p$, we will show that \begin{align}
I(\alpha,\beta)
=
\tilde{I}(\alpha,\beta)+o(\delta^2)
\label{e245}
\end{align}
as $\delta\to 0$, where we defined
\begin{align}
\tilde{I}(\alpha,\beta):=\frac{\delta^2}{2\bar{p}^4}\alpha\bar{\alpha}\beta\bar{\beta}.
\label{e_itilde}
\end{align}
If the difference between $I(\alpha,\beta)$ and $\tilde{I}(\alpha,\beta)$ could be neglected, then \eqref{e_i2d} should hold.
To see \eqref{e245}, for given $\alpha,\beta\in (0,1)$, 
note that \eqref{e_scaling} implies,
\begin{align}
\frac{(1+\frac{\lambda}{\alpha\beta})(1+\frac{\lambda}{\bar{\alpha}\bar{\beta}})}
{(1-\frac{\lambda}{\alpha\bar{\beta}})(1-\frac{\lambda}{\bar{\alpha}\beta})}
=
\frac{(1+\delta)(1+\frac{\delta p^2}{\bar{p}^2})}
{(1-\frac{\delta p}{\bar{p}})^2}.
\label{e_blambda}
\end{align}
Under the assumption $\delta\to 0$,
the above linearizes to 
\begin{align}
\frac{\lambda}{\alpha\bar{\alpha}\beta\bar{\beta}}
=
\frac{\delta}{\bar{p}^2}+o(\delta).
\label{e_inv}
\end{align}
Moreover, note that
\begin{align}
D_{\chi^2}
(P^{\alpha,\beta}\|P^{\alpha}\times P^{\beta})
&=
\frac{\lambda^2}{2\alpha\bar{\alpha}\beta\bar{\beta}}
\\
&=\frac{\delta^2}{2\bar{p}^4}\alpha\bar{\alpha}\beta\bar{\beta}
+o(\delta^2)
\end{align}
where the last step follows by comparing with \eqref{e_inv}. 
Since $I(\alpha,\beta)/D_{\chi^2}
(P^{\alpha,\beta}\|P^{\alpha}\times P^{\beta})\to 1$ as $\delta\to 0$, we see \eqref{e245} holds.
Of course, \eqref{e245} does not really show \eqref{e_i2d} since approximation of function values generally does not imply approximation of the derivatives.
However, we shall next explicitly take the derivatives to give a real proof, and the above observations are useful guides.

First, note that
\begin{align}
&\frac{\partial I(\alpha,\beta)}{\partial \alpha}
\nonumber\\
&=\sum_{x,y\in\{0,1\}}\left(
\frac{\partial}{\partial\alpha}P^{\alpha,\beta}(x,y)
\right)
\ln\frac{P^{\alpha,\beta}(x,y)}{P^{\alpha}(x)P^{\beta}(x)}
\\
&=(\beta+\lambda_{\alpha})\ln(1+\frac{\lambda}{\alpha\beta})
+(\bar{\beta}-\lambda_{\alpha})\ln(1-\frac{\lambda}{\alpha\bar{\beta}})\nonumber
\\
&\quad (-\beta-\lambda_{\alpha})\ln(1-\frac{\lambda}{\bar{\alpha}\beta})
+(-\bar{\beta}+\lambda_{\alpha})\ln(1+\frac{\lambda}{\bar{\alpha}\bar{\beta}}).
\label{e_ialpha}
\end{align}
where $\lambda_{\alpha}:=\frac{\partial}{\partial\alpha}\lambda$.
Next, we express the first and second derivatives, $\lambda_{\alpha}$, $\lambda_{\beta}$ and $\lambda_{\alpha,\beta}$, in terms of $\lambda$.
Differentiating the logarithm of \eqref{e_blambda} in $\beta$ yields 
\begin{align}
&\quad\lambda_{\beta}\left[
\frac1{\alpha\beta+\lambda}
+\frac1{\bar{\alpha}\bar{\beta}+\lambda}
+\frac1{\alpha\bar{\beta}-\lambda}
+\frac1{\bar{\alpha}\beta-\lambda}
\right]
\nonumber\\
&=\lambda\left[
\frac{\beta^{-1}}{\alpha\beta+\lambda}
-\frac{\bar{\beta}^{-1}}{\bar{\alpha}\bar{\beta}+\lambda}
-\frac{\bar{\beta}^{-1}}{\alpha\bar{\beta}-\lambda}
+\frac{\beta^{-1}}{\bar{\alpha}\beta-\lambda}
\right].
\label{e168}
\end{align}
In the rest of the proof the notation $f(t)=O(t)$ means $|f(t)|\lesssim |t|$
(recall the definition of $\lesssim$ in \eqref{e_i2d}), and $f(t)=\Theta(t)$ if $1\lesssim f(t)/t\lesssim 1$.
Note that for $0.1p<\alpha,\beta<10p$ we have
\begin{align}
\frac{\lambda}{\alpha\beta}=\Theta(\delta),
\end{align}
since the right side of \eqref{e_blambda} clearly equals $1+\Theta(\delta)$.
Then by \eqref{e168},
\begin{align}
\lambda_{\beta}\left[
\frac1{\alpha\bar{\alpha}\beta\bar{\beta}}
+O(\frac{\lambda}{\alpha^2\beta^2})
\right]
=\lambda\left[\frac{\bar{\beta}-\beta}{\alpha\bar{\alpha}\beta^2\bar{\beta}^2}
+O(\frac{\lambda}{\alpha^2\beta^3})
\right]
\end{align}
and hence,
\begin{align}
\lambda_{\beta}=\lambda\cdot\frac{\bar{\beta}-\beta}{\beta\bar{\beta}}\left(1+O(\frac{\lambda}{\alpha\beta})\right)
=O(\frac{\lambda}{\beta}).
\label{e167}
\end{align}
Expression of $\lambda_{\alpha}$ can be found similarly.
Moreover, differentiating \eqref{e168} we get
\begin{align}
&\lambda_{\alpha,\beta}\left[
\tfrac1{\alpha\beta+\lambda}
+\tfrac1{\bar{\alpha}\bar{\beta}+\lambda}
+\tfrac1{\alpha\bar{\beta}-\lambda}
+\tfrac1{\bar{\alpha}\beta-\lambda}
\right]
\nonumber\\
&+\lambda_{\alpha}\lambda_{\beta}
\left[
-\tfrac1{(\alpha\beta+\lambda)^2}
-\tfrac1{(\bar{\alpha}\bar{\beta}+\lambda)^2}
+\tfrac1{(\alpha\bar{\beta}-\lambda)^2}
+\tfrac1{(\bar{\alpha}\beta-\lambda)^2}
\right]
\nonumber\\
&+\lambda_{\alpha}\left[
-\tfrac{\alpha}{(\alpha\beta+\lambda)^2}
+\tfrac{\bar{\alpha}}{(\bar{\alpha}\bar{\beta}+\lambda)^2}
+\tfrac{\alpha}{(\alpha\bar{\beta}-\lambda)^2}
-\tfrac{\bar{\alpha}}{(\bar{\alpha}\beta-\lambda)^2}
\right]
\nonumber\\
&+\lambda_{\beta}\left[
-\tfrac{\beta}{(\alpha\beta+\lambda)^2}
+\tfrac{\bar{\beta}}{(\bar{\alpha}\bar{\beta}+\lambda)^2}
+\tfrac{\beta}{(\bar{\alpha}\beta-\lambda)^2}
-\tfrac{\bar{\beta}}{(\alpha\bar{\beta}-\lambda)^2}
\right]
\nonumber\\
&+\lambda\left[\tfrac1{(\alpha\beta+\lambda)^2}
+\tfrac1{(\bar{\alpha}\bar{\beta}+\lambda)^2}
-\tfrac1{(\alpha\bar{\beta}-\lambda)^2}
-\tfrac1{(\bar{\alpha}\beta-\lambda)^2}
\right]=0,
\end{align}
from which we can deduce that 
\begin{align}
\lambda_{\alpha,\beta}=O(\frac{\lambda}{\alpha\beta}).
\label{e169}
\end{align}
Now, taking the derivative in \eqref{e_ialpha}, we obtain
\begin{align}
&\quad\partial_{\alpha,\beta}I(\alpha,\beta)
\nonumber\\
&=
\lambda_{\alpha,\beta}\lambda\cdot\frac1{\alpha\bar{\alpha}\beta\bar{\beta}}
+\lambda_{\alpha}\lambda_{\beta}\cdot\frac1{\alpha\bar{\alpha}\beta\bar{\beta}}
+\lambda_{\alpha}\lambda\cdot\frac{\beta-\bar{\beta}}{\alpha\bar{\alpha}\beta^2\bar{\beta}^2}
\nonumber\\
&\quad +\lambda_{\beta}\lambda\cdot\frac{\alpha-\bar{\alpha}}{\alpha^2\bar{\alpha}^2\beta\bar{\beta}}
+\frac{\lambda^2}{2}\cdot\frac{(\bar{\alpha}-\alpha)(\bar{\beta}-\beta)}{\alpha^2\bar{\alpha}^2\beta^2\bar{\beta}^2}
\nonumber\\
&\quad+O\left(\frac{\lambda^3}{\alpha^3\beta^3}\right).
\label{e164}
\end{align}
In deriving \eqref{e164}, we applied the Taylor expansions of $x\mapsto \ln(1+x)$ and  $x\mapsto \frac1{1+x}$.
Plugging \eqref{e167} and \eqref{e169} into 
\eqref{e164}, we obtain 
\begin{align}
|\partial_{\alpha,\beta}I(\alpha,\beta)|=O(\left(\frac{\lambda}{\alpha\beta}\right)^2)=O(\delta^2).
\label{e229}
\end{align}

Next, we control $|\partial_{\alpha,\alpha}I(\alpha,\beta)|$.
Similarly to \eqref{e168}, we have
\begin{align}
&\quad\lambda_{\alpha}\left[
\frac1{\alpha\beta+\lambda}
+\frac1{\bar{\alpha}\bar{\beta}+\lambda}
+\frac1{\bar{\alpha}\beta-\lambda}
+\frac1{\alpha\bar{\beta}-\lambda}
\right]
\nonumber\\
&=\lambda\left[
\frac{\alpha^{-1}}{\alpha\beta+\lambda}
-\frac{\bar{\alpha}^{-1}}{\bar{\alpha}\bar{\beta}+\lambda}
-\frac{\bar{\alpha}^{-1}}{\bar{\alpha}\beta-\lambda}
+\frac{\alpha^{-1}}{\alpha\bar{\beta}-\lambda}
\right].
\end{align}
Further taking the derivative, we obtain
\begin{align}
&\quad\lambda_{\alpha,\alpha}\left[
\frac1{\alpha\beta+\lambda}
+\frac1{\bar{\alpha}\bar{\beta}+\lambda}
+\frac1{\bar{\alpha}\beta-\lambda}
+\frac1{\alpha\bar{\beta}-\lambda}
\right]
\nonumber
\\
&+\lambda_{\alpha}
\left[
\frac{-2\beta-\alpha^{-1}\lambda}{(\alpha\beta+\lambda)^2}
+\frac{2\bar{\beta}+\bar{\alpha}^{-1}\lambda}{(\bar{\alpha}\bar{\beta}+\lambda)^2}
\right.
\nonumber\\
&\quad\quad+\left.\frac{2\beta-\bar{\alpha}^{-1}\lambda}{(\bar{\alpha}\beta-\lambda)^2}
+\frac{-2\bar{\beta}+\alpha^{-1}\lambda}{(\alpha\bar{\beta}-\lambda)^2}
\right]
\nonumber
\\
&+\lambda^2\left[
\frac1{\alpha^2(\alpha\beta+\lambda)^2}
+\frac1{\bar{\alpha}^2(\bar{\alpha}\bar{\beta}+\lambda)^2}
\right.
\nonumber\\
&\quad\quad\left.
-\frac1{\bar{\alpha}^2(\bar{\alpha}\beta-\lambda)^2}
-\frac1{\alpha^2(\alpha\bar{\beta}-\lambda)^2}
\right]
\nonumber\\
&+2\lambda
\left[\frac{\beta}{\alpha(\alpha\beta+\lambda)^2}
+\frac{\bar{\beta}}{\bar{\alpha}(\bar{\alpha}\bar{\beta}+\lambda)^2}
\right.
\nonumber\\
&\quad\quad\left.
+\frac{\beta}{\bar{\alpha}(\bar{\alpha}\beta-\lambda)^2}
+\frac{\bar{\beta}}{\alpha(\alpha\bar{\beta}-\lambda)^2}
\right]
=0.\label{e230}
\end{align}
Next, we shall use the assumption of $\alpha,\beta\in(0.1p,10p)$ to simplify \eqref{e230} as
\begin{align}
\lambda_{\alpha,\alpha}\cdot\Theta(\frac1{p^2})
-\lambda_{\alpha}\cdot \Theta(\frac1{p^3})
+\lambda^2\cdot \Theta(\frac1{p^6})
+\lambda\cdot \Theta(\frac1{p^4})
=0.\label{e231}
\end{align}
Since, analogous to \eqref{e167}, we have 
\begin{align}
\lambda_{\alpha}=O(\frac{\lambda}{\alpha}),
\label{e_lalpha}
\end{align}
we see that \eqref{e231} implies
\begin{align}
\lambda_{\alpha,\alpha}=O(\frac{\lambda}{p^2}).
\end{align}
Tighter estimates of $\lambda_{\alpha,\alpha}$ are possible, but the above will suffice.
We now take the derivative of \eqref{e_ialpha} in $\alpha$:
\begin{align}
\partial_{\alpha,\alpha}I(\alpha,\beta)
=I_1+I_2
\end{align}
where
\begin{align}
I_1:&=\lambda_{\alpha,\alpha}\left[\ln(1+\frac{\lambda}{\alpha\beta})
-\ln(1-\frac{\lambda}{\alpha\bar{\beta}})
\right.
\nonumber\\
&\quad\quad
\left.
-\ln(1-\frac{\lambda}{\bar{\alpha}\beta})
+\ln(1+\frac{\lambda}{\bar{\alpha}\bar{\beta}})
\right]
\end{align}
and
\begin{align}
I_2&:=\lambda_{\alpha}^2\left(
\frac{\frac1{\alpha\beta}}{1+\frac{\lambda}{\alpha\beta}}
+\frac{\frac1{\alpha\bar{\beta}}}{1-\frac{\lambda}{\alpha\bar{\beta}}}
+\frac{\frac1{\bar{\alpha}\beta}}{1-\frac{\lambda}{\bar{\alpha}\beta}}
+\frac{\frac1{\bar{\alpha}\bar{\beta}}}{1+\frac{\lambda}{\bar{\alpha}\bar{\beta}}}
\right)
\nonumber\\
&+\lambda_{\alpha}\left(
\frac{\frac1{\alpha}}{1+\frac{\lambda}{\alpha\beta}}
-\frac{\frac1{\alpha}}{1-\frac{\lambda}{\alpha\bar{\beta}}}
+\frac{\frac1{\bar{\alpha}}}{1-\frac{\lambda}{\bar{\alpha}\beta}}
-\frac{\frac1{\bar{\alpha}}}{1+\frac{\lambda}{\bar{\alpha}\bar{\beta}}}
\right)
\nonumber\\
&+\lambda_{\alpha}\lambda\left(
\frac{-\frac1{\alpha^2\beta}}{1+\frac{\lambda}{\alpha\beta}}
+\frac{-\frac1{\alpha^2\bar{\beta}}}{1-\frac{\lambda}{\alpha\bar{\beta}}}
+\frac{\frac1{\bar{\alpha}^2\beta}}{1-\frac{\lambda}{\bar{\alpha}\beta}}
+\frac{\frac1{\bar{\alpha}^2\bar{\beta}}}{1+\frac{\lambda}{\bar{\alpha}\bar{\beta}}}
\right)
\nonumber\\
&+\lambda\left(
\frac{-\frac1{\alpha^2}}{1+\frac{\lambda}{\alpha\beta}}
+\frac{\frac1{\alpha^2}}{1-\frac{\lambda}{\alpha\bar{\beta}}}
+\frac{\frac1{\bar{\alpha}^2}}{1-\frac{\lambda}{\bar{\alpha}\beta}}
+\frac{-\frac1{\bar{\alpha}^2}}{1+\frac{\lambda}{\bar{\alpha}\bar{\beta}}}
\right)
\end{align}
We can Taylor expand $I_1$ using the facts that $\lambda_{\alpha,\alpha}=O(\frac{\lambda}{p^2})$, $\alpha,\beta=\Theta(p)$, to obtain
\begin{align}
I_1=O(\frac{\lambda^2}{p^4}).
\end{align}
We can Taylor expand $I_2$ using the fact that $\lambda_{\alpha}=O(\frac{\lambda}{p})$ (analogous to \eqref{e167}) to obtain
\begin{align}
I_2=O(\frac{\lambda^2}{p^4}).
\end{align}
Thus
\begin{align}
|\partial_{\alpha,\alpha}I(\alpha,\beta)|=O(\frac{\lambda^2}{p^4})=O(\delta^2).
\end{align}
By symmetry same bound holds for $|\partial_{\beta,\beta}I(\alpha,\beta)|$ as well.
Together with \eqref{e229},
we thus validated \eqref{e_i2d}, and consequently \eqref{e_suff} in this case.

\item {\bf Case~2:} $\max\{\alpha,\beta\}\ge 10p$.\\
Without loss of generality assume that $\alpha\ge\beta$ and $\alpha\ge 10p$.
From \eqref{e_ialpha}, we have
\begin{align}
&\quad\partial_{\alpha}I(p,p)
\nonumber\\
&=(p+\lambda_{\alpha}(p,p))\ln(1+\delta)
+(\bar{p}-\lambda_{\alpha}(p,p))\ln(1-\delta p/\bar{p})
\nonumber\\
&\quad+(-p-\lambda_{\alpha}(p,p))\ln(1-\delta p/\bar{p})
\nonumber\\
&\quad+(-\bar{p}+\lambda_{\alpha}(p,p))\ln(1+\delta p^2/\bar{p}^2).
\end{align}
Using \eqref{e_lalpha} with $\lambda\leftarrow \delta p^2$ and $\alpha\leftarrow p$, we obtain $\lambda_{\alpha}(p,p)=O(p\delta)$.
Thus Taylor expanding the above displayed, we obtain
\begin{align}
\partial_{\alpha}I(p,p)\lesssim p\delta^2.
\label{e_pibound}
\end{align}
Then 
\begin{align}
\phi(\alpha,\beta)&:=I(p,p)-I(\alpha,\beta)+I_{\alpha}(p,p)(\alpha-p)
\nonumber\\
&\quad
+I_{\beta}(p,p)(\beta-p)
\\
&\le p^2\delta^2-0+2I_{\alpha}(p,p)(\alpha-p)
\\
&\lesssim p\alpha\delta^2
\label{e_phib}
\end{align}
where we used the assumption that $\alpha\ge\beta$ and the fact that $I_{\alpha}(p,p)=I_{\beta}(p,p)$.
Now the assumption of 
$\alpha\ge 10p$ implies 
\begin{align}
\psi(\alpha,\beta)\ge d(\alpha\|p)\gtrsim \alpha.
\label{e235}
\end{align}
To see the second inequality in \eqref{e235}, note that 
\begin{align}
\min_{\alpha\in(10p,1]}\frac1{\alpha}d(\alpha\|p)
&=\min_{\alpha\in[10p,1]}\left\{
\ln\frac{\alpha}{p}+\frac{1-\alpha}{\alpha}\ln\frac{1-\alpha}{1-p}
\right\}
\\
&=\frac{d(10p\|p)}{10p}
\end{align}
where the minimization is easily solved by checking that the derivative is positive for $\alpha\ge 10p$.
Finally, combining \eqref{e235} with \eqref{e_phib},
we obtain $\frac{\phi(\alpha,\beta)}{\psi(\alpha,\beta)}\lesssim \delta^2p$, as desired.

\item {\bf Case~3:} $\min\{\alpha,\beta\}\le 0.1p$, $\max\{\alpha,\beta\}<10p$.\\
Assume without loss of generality that $\alpha\le 0.1p$.
In this case, using \eqref{e_pibound}, we have
\begin{align}
\phi(\alpha,\beta)&:=I(p,p)-I(\alpha,\beta)+I_{\alpha}(p,p)(\alpha-p)
\nonumber\\
&\quad
+I_{\beta}(p,p)(\beta-p)
\\
&\le p^2\delta^2-0+I_{\alpha}(p,p)(\alpha+\beta-2p)
\\
&\le  p^2\delta^2+I_{\alpha}(p,p)\cdot 18p
\\
&=O(p^2\delta^2).
\end{align}
On the other hand,
\begin{align}
\psi(\alpha,\beta)
&\ge d(\alpha\|p)
\\
&\ge d(0.1p\|p)
\\
&=(0.9-0.1\ln10)p+O(p^2)
\\
&=\Theta(p).
\end{align}
Thus we once again obtain $\frac{\phi(\alpha,\beta)}{\psi(\alpha,\beta)}\lesssim \delta^2p$, as desired.
\end{itemize}

\bibliographystyle{plainurl}
\bibliography{ref_ccnp}

\begin{IEEEbiographynophoto}{Jingobo Liu}
received the B.S. degree in Electrical Engineering from Tsinghua University, Beijing, China in 2012, and the M.A. and Ph.D. degrees in Electrical Engineering from Princeton University, Princeton, NJ, USA, in 2014 and 2017. He was a Norbert Wiener Postdoctoral Research Fellow at the MIT Institute for Data, Systems, and Society (IDSS) during 2018-2020.
Since 2020, he has been an assistant professor in the Department of Statistics and an affiliate in the Department of Electrical and Computer Engineering at the University of Illinois, Urbana-Champaign, IL, USA.

His research interests include information theory, high dimensional statistics and probability, coding theory, and the related fields. His undergraduate thesis received the best undergraduate thesis award at Tsinghua University (2012). He gave a semi-plenary presentation at the 2015 IEEE Int. Symposium on Information Theory, Hong-Kong, China. He was a recipient of the Princeton University Wallace Memorial Honorific Fellowship in 2016. His Ph.D. thesis received the Bede Liu Best Dissertation Award of Princeton and the Thomas M. Cover Dissertation Award of the IEEE Information Theory Society (2018).
\end{IEEEbiographynophoto}
\end{document}